\documentclass[onefignum,onetabnum]{siamonline171218}

\usepackage{amsmath,amssymb,mathrsfs,graphicx,subfig,hyperref,xcolor}

\usepackage{enumitem}
\setlist[enumerate]{leftmargin=.5in}
\setlist[itemize]{leftmargin=.5in}

\usepackage{amsopn}

\newcommand{\gen}{\mathscr{L}}

\DeclareMathOperator{\diam}{diam}

\newcommand{\ignore}[1]{}
\newcommand{\oldstuff}[1]{}
\newcommand{\E}{\mathrm{\mathbf{E}}}
\newcommand{\1}{\mathrm{\mathbf{1}}}
\renewcommand{\P}{\mathrm{\mathbf{P}}}

\DeclareMathOperator{\var}{var}
\DeclareMathOperator{\cov}{cov}

\DeclareMathOperator{\acov}{C}
\newcommand{\convdist}{\xrightarrow[]{\rm{d}}}

\newcommand{\convas}{\xrightarrow[]{\rm{as}}}

\newcommand{\N}{\rm{N}} 

\newcommand{\iidx}{\mathbf{i}}
\newcommand{\jidx}{\mathbf{j}}
\newcommand{\kidx}{\mathbf{k}}

\newcommand{\Real}{\mathbb{R}}
\newcommand{\Z}{\mathbb{Z}}
\newcommand{\eps}{\varepsilon}

\DeclareMathOperator{\diver}{div}

\DeclareMathOperator*{\diag}{diag}

\renewcommand{\t}{\mathrm{t}}

\DeclareMathOperator*{\tr}{tr}

\newcommand{\zv}{z^{\rm V}}

\newcommand{\emus}{\text{US}}
\newcommand{\pu}[1]{#1^\ast}

\newsiamremark{remark}{Remark}
\newsiamremark{assumption}{Assumption}

\headers{Stratification and Markov chain Monte Carlo}{Dinner, et al.\@}

\title{
  Stratification as a general variance reduction method for Markov chain Monte Carlo
  \thanks{BvK was supported by NSF RTG: Computational and Applied Mathematics in Statistical Science, number 1547396. ARD, EHT, and JW were supported by National Institutes of Health (NIH) Grant Number R01GM109455.  JW was also supported by the Advanced Scientific Computing Research Program within the DOE Office of Science through award DE-SC0020427. Computing resources were provided by the University of Chicago Research Computing Center.  We wish to thank Jonathan Mattingly, Jeremy Tempkin, and Charlie Matthews for helpful discussions.}
}

\author{
  Aaron R.\@ Dinner \thanks{Department of Chemistry and James Franck Institute, the University of Chicago, Chicago, IL 60637}
  \and Erik H.\@ Thiede \thanks{Department of Chemistry and James Franck Institute, the University of Chicago, Chicago, IL 60637}
  \and Brian Van Koten \thanks{Department of Mathematics and Statistics, the University of Massachusetts, Amherst, MA 01003}
  \and Jonathan Weare \thanks{Courant Institute of Mathematical Sciences, New York University, New York, NY 10012}}

\begin{document}

\maketitle

\begin{abstract}
  The Eigenvector Method for Umbrella Sampling (EMUS) \cite{Thiede2016} belongs to a popular class of methods in statistical mechanics which adapt the principle of stratified survey sampling to the computation of free energies.
  We develop a detailed theoretical analysis of EMUS.
  Based on this analysis, we show that EMUS is an efficient \emph{general} method for computing averages over arbitrary target distributions. In particular, we show that EMUS can be dramatically more efficient than direct MCMC when the target distribution is multimodal or when the goal is to compute tail probabilities. To illustrate these theoretical results, we present a tutorial application of the method to a problem from Bayesian statistics. 
\end{abstract}

\section{Introduction}
Markov chain Monte Carlo (MCMC) methods have been widely used with great success throughout statistics, engineering, and the natural sciences. 
However, when estimating tail probabilities or when sampling from multimodal distributions, accurate MCMC estimates often require a prohibitively large number of samples. 
In this article, we analyze the Eigenvector Method for Umbrella Sampling (EMUS)~\cite{Thiede2016}.  We first proposed EMUS as a method for computing free energies, and we demonstrated that it was useful for treating the multimodality that typically arises in that context.
Here, we demonstrate that EMUS is an effective \emph{general} means of addressing the challenges posed not just by multimodality but also tail events, with potential applications to a broad range of problems in statistics, engineering, and the natural sciences.

EMUS was inspired by Umbrella Sampling~\cite{torrie1977nonphysical} and other methods such as the Weighted Histogram Analysis Method (WHAM)~\cite{WHAM1992} and the Multistate Bennett Acceptance Ratio (MBAR)~\cite{shirts2008statistically} for computing potentials of mean force and free energies in statistical mechanics. 
\footnote{A potential of mean force is the logarithm of a marginal density. A free energy is the logarithm of a normalization constant. Both quantities play fundamental roles in statistical mechanics, e.g.,\@ in theories of rates of chemical reactions.} We call these \emph{stratified MCMC methods} since they each adapt the principle of stratified survey sampling to MCMC simulation. That is,  they each draw samples concentrated in a collection of subregions of state space and combine the resulting averages in a consistent manner.
Stratified MCMC methods are among the most powerful, most successful, and most widely used tools in molecular simulation. (However, in contrast to our presentation here, they are not typically used in molecular simulation to compute averages of general observables.) WHAM, for example, has been instrumental for treating biomolecular processes ranging from  protein folding~\cite{boczko1995folding} to conductance by ion channels~\cite{berneche2001energetics}.

While the practical utility of stratification has been established in many applications, the advantages and disadvantages of the method have remained poorly understood; cf.\@~\cite{Thiede2016}. 
Motivated by the substantial gap between theory and application of stratified MCMC within statistical mechanics, and also by the general challenges posed by multimodality and tail probabilities, the goal of this paper is to develop a clear theoretical explanation of the advantages of EMUS. Our theory suggests new applications of stratified MCMC (and EMUS in particular) to broad classes of sampling problems arising in statistics and statistical mechanics. For example, very recently EMUS was successfully applied to a parameter estimation problem in cosmology~\cite{matthews2017cosmo}.


\paragraph{Summary of Main Results}
Our most general results are a central limit theorem (CLT) for the EMUS method and a convenient upper bound on the asymptotic variance, cf.\@ Theorem~\ref{thm: CLT for EMUS} and Theorem~\ref{thm: upper bound on asymptotic variance}. We note that the proof of the upper bound relies on a new class of perturbation estimates for Markov chains which we derived in~\cite{ThVKWe:Perturbation}. These estimates are substantially more detailed than previous results~\cite{ChoMey:Survey}. 
After proving the CLT, we address the dependence of the sampling error on the choice of strata. In particular, for a representative MCMC method, we estimate the asymptotic variances of trajectory averages sampling the biased distributions, cf.\@ Theorem~\ref{thm: ergodicity and estimates of iat}. Our estimate shows how factors such as the diameters of the strata influence the asymptotic variances.  

In Section~\ref{sec: limiting results}, we apply the general theory developed in Section~\ref{sec: error analysis} to case studies involving tail probabilities and multimodality. Our results concern two limits: a \emph{small probability} limit and a \emph{low-temperature} limit. In the small probability limit, we consider estimation of probabilities of the form
\begin{equation*}
p_M := \P[X \geq M].
\end{equation*}
For a broad class of random variables $X$, we show that while the cost of computing $p_M$ with relative precision by direct MCMC increases exponentially with $M$, the cost by EMUS increases only polynomially; cf.\@ Section~\ref{sec: small probability limit}.
In the low-temperature limit, a parameter of the target distribution decreases, intensifying the effects of multimodality on the efficiency of MCMC sampling.
We show that the cost of computing an average to fixed precision by direct MCMC increases exponentially in this limit, whereas the cost by EMUS increases only polynomially; cf.\@ Section~\ref{sec: low temp}. 
We conclude that EMUS may be dramatically more efficient than direct MCMC sampling when the target distribution is multimodal or when the goal is to compute a small tail probability.

To illustrate our theoretical results, we present a tutorial numerical study applying EMUS to a problem in Bayesian statistics in Section~\ref{sec: numerical results}. 
In addition to illustrating the theory, our numerical study demonstrates the problems that may occur when EMUS and other similar stratified MCMC methods are used carelessly. It also addresses practical issues such as the choice of strata and the computation of error bars for averages estimated by EMUS.

The results in this article significantly extend and generalize the ideas in~\cite{Thiede2016}. We first proposed the EMUS method with the goal of analyzing and improving umbrella sampling approaches in free energy calculations.
Here, our goal is to establish EMUS as a \emph{general} variance reduction technique, and we present entirely new results, including an upper bound on the asymptotic variance of EMUS (Theorem~\ref{thm: upper bound on asymptotic variance}), a condition to guide some aspects of the choice of strata (Remark~\ref{rem: remark on size of entries in F and choice of strata}), a theoretical argument demonstrating the benefits of EMUS for computing tail probabilities (Section~\ref{sec: small probability limit}), numerical results applying EMUS to Bayesian inference (Section~\ref{sec: numerical results}), a method of correcting problems related to poorly chosen strata (Section~\ref{subsec: numerical experiments}), and a greatly improved numerical method for estimating the standard deviations of quantities computed by EMUS (Appendix~\ref{sec:error-bars}).
In addition, we give complete justifications of some results that were stated without proof in~\cite{Thiede2016}, including Theorem~\ref{thm: ergodicity and estimates of iat} concerning the dependence of the sampling error on the choice of strata.
Finally, we note that our results concerning multimodal distributions and the low-temperature limit generalize and clarify the results given in~\cite{Thiede2016}; in particular, our Theorem~\ref{thm: low temp EMUS} covers periodic boundary conditions and stratification in more than one variable.

\section{The Eigenvector Method for Umbrella Sampling}\label{sec: EMUS}
In this section, we present the Eigenvector Method for Umbrella Sampling (EMUS), 
and we prove that it is consistent.  A detailed derivation of the estimator can be found in \cite{Thiede2016}.
We also review a related method, iterative EMUS, and we compare iterative EMUS with the MBAR method from statistical mechanics~\cite{shirts2008statistically}.

\subsection{The EMUS estimator}\label{sec:emus}
The objective of EMUS is to compute the average 
\begin{equation*}
 \pi[g]:= \int_\Omega g(x) \pi(dx), 
\end{equation*}
of a function $g$ with respect to a measure $\pi$ defined on a set $\Omega$.
In EMUS, instead of sampling directly from $\pi$, we sample from biased distributions analogous to the strata in stratified survey sampling methods. 
We then weight the samples from the biased distributions to estimate $\pi[g]$.

We assume that the \emph{biased distributions} take the form 
\begin{equation*}
 \pi_i(dx) :=  \frac{\psi_i(x) \pi(dx)}{\pi[\psi_i]}
\end{equation*}
for some set $\{\psi_i\}_{i=1}^L$ of non-negative $\emph{bias functions}$ defined on $\Omega$.
We assume that 
\begin{equation*} 
 \sum_{i=1}^L \psi_i(x) >0 \text{ for all } x \in \Omega.
\end{equation*}
We call the support of $\psi_i$ the \emph{$i$'th stratum} to make an analogy between the biased distributions of EMUS and the strata of stratified survey sampling.

The EMUS estimator is based on two observations. First, one may write $\pi[g]$ in terms of weighted sums of averages over the biased distributions.  Let $u\in (0,\infty)^L$ be arbitrary, and for any function $h:\Omega \rightarrow \Real$, define
\begin{equation*}
 h^\ast (x) := \frac{h(x)}{\sum_{k=1}^L \psi_k(x)/ u_k}.
\end{equation*}
By~\cite[Equation~16]{Thiede2016}, we have 
\begin{equation}\label{eq: decomposition of average 2}
 \pi[g] = \frac{\sum_{i=1}^L z_i \,\pi_i [ \pu{g} ] / u_i}{\sum_{i=1}^L z_i \,\pi_i [\1^\ast]/u_i},
\end{equation}
where $\1$ denotes the function equal to one everywhere and
\begin{equation}\label{eq: first formula for weights}
z_i = \frac{\pi[\psi_i]}
{\sum_{k=1}^L \pi[\psi_k]}.
\end{equation}
(For the reader's convenience, we present complete derivations of equations~\eqref{eq: decomposition of average 2} and~\eqref{eq: eigenvector problem} in Appendix~\ref{apx: derivation of emus}.) The parameter $u$ above can be thought of as an initial guess of the unknown \emph{weight vector} $z \in \Real^L$ in~\eqref{eq: first formula for weights}.  Its explicit presence here will simplify the description of an iterative version of EMUS in Section~\ref{sec: iterative EMUS}. 
Outside of this section and Section~\ref{sec: iterative EMUS}, the choice of $u$ is absorbed in the definition of the $\psi_i$, i.e. $u_i=1$.

Second, the weight vector $z$ can be found by solving a certain eigenproblem. 
Let $w \in \Real^L$ be defined by $z_i = u_i w_i$. 
By~\cite[Equation~17]{Thiede2016}, we have  
\begin{equation}\label{eq: eigenvector problem}
w^\t F=w^\t,  \quad \text{ where } F_{ij} := \pi_i [\pu{\psi_j}]/u_j.
\end{equation}
We call $F$ the \emph{overlap matrix}.

We observe that equations~\eqref{eq: decomposition of average 2} and~\eqref{eq: eigenvector problem} combine to express $\pi[g]$ as a function of averages over the biased distributions, namely $\pi_i[g^\ast]$, $\pi_i[\1^\ast]$, and $F$. To see this, it suffices to recognize that $F$ is stochastic and $w$ is a probability vector. Therefore, whenever $F$ is irreducible, the solution of the eigenproblem is unique by the Perron--Frobenius theorem. We will assume throughout the remainder of this work that $F$ is irreducible, which amounts to requiring some overlap between neighboring strata; see Lemma~\ref{lem: irreducibility condition}. In general, for any irreducible, stochastic matrix $G \in \Real^{L \times L}$, we will let $w(G) \in \Real^L$ denote the unique solution of 
\begin{equation*}
 w(G)^\t G = w(G)^\t \text{ with } \sum_{i=1}^L w_i(G) = 1.
\end{equation*}
With this notation, by~\eqref{eq: decomposition of average 2} and~\eqref{eq: eigenvector problem} we have 
\begin{equation}\label{eq: final decomposition formula}
 \pi[g] = \frac{\sum_{i=1}^L w_i(F) \pi_i [g^\ast]}{\sum_{i=1}^L w_i(F) \pi_i [\1^\ast]}.
\end{equation}

In EMUS, we substitute MCMC estimates for the averages over biased distributions on the right hand side of~\eqref{eq: final decomposition formula} to estimate $\pi[g]$. 
To be precise, let $X^i_t$ be a Markov process ergodic for $\pi_i$.
We call $X^i_t$ the \emph{biased process} sampling the biased distribution $\pi_i$.
The EMUS algorithm proceeds as follows:
\begin{enumerate}
 \item For each $i = 1, \dots, L$, compute $N_i$ steps of the process $X^i_t$.
 \item Compute the averages
 \begin{equation}\label{eq:barg1F}
  \bar g^\ast_i := \frac{1}{N_i} \sum_{t=1}^{N_i} g^\ast(X^i_t), \text{ } 
  \bar \1^\ast_i := \frac{1}{N_i} \sum_{t=1}^{N_i} \1^\ast(X^i_t), \text{ and } 
  \bar F_{ij} := \frac{1}{N_i} \sum_{t=1}^{N_i} \psi^\ast_j(X^i_t)/u_j.
 \end{equation}
 \item Compute $w(\bar F)$ numerically, for example from the QR factorization of $I-\bar F$ \cite{GolMey:ComputingInvDist}.
 \item Compute the estimate 
 \begin{equation*}
  \pi_\emus[g] := \frac{ \sum_{i=1}^L w_i(\bar F) \bar g^\ast_i}{\sum_{i=1}^L w_i(\bar F) \bar \1^\ast_i}
 \end{equation*}
 of $\pi[g]$.
\end{enumerate}

Recall that $w(\bar F)$ is defined only if $\bar F$ is irreducible.
In the following lemma, whose proof is in Appendix~\ref{apx: irreducibility condition}, we state simple criteria for the irreducibility of $F$ and $\bar F$:

\begin{lemma}\label{lem: irreducibility condition}
The overlap matrix $F$ is irreducible if and only if for every $A \subset \{1,2,\dots,L\}$, we have
\begin{equation}\label{eq: irreducibility condition}
\pi \left [ \left (\sum_{i \in A} \psi_i \right ) \left (\sum_{j \notin A} \psi_j \right ) \right ] >0.
\end{equation}
The approximate overlap matrix $\bar F$ is irreducible if and only if for every $A \subset \{1,2,\dots,L\}$, the set
$
\cup_{i \in A} \{ x: \psi_i(x) >0 \}
$
contains at least one sample point generated from one of the biased processes $X^j_t$ with $j \notin A$.
\end{lemma}

We claim that the EMUS estimator is consistent; that is, $\pi_\emus[g]$ converges almost surely to $\pi[g]$ as the total number of samples tends to infinity.
To make this precise, we require the following assumption on the growth of $N_i$ with the total number of samples:
\begin{assumption}\label{asm: defn of kappa}
Let
\begin{equation*}
  N= \sum_{i=1}^L N_i
\end{equation*}
be the total number of samples from all biased distributions.
Assume that for each $i$,
\begin{equation*}
  \lim_{N \rightarrow \infty} N_i/N = \kappa_i >0.
\end{equation*}
That is, assume that when $N$ is large, the proportion of samples drawn from the $i$'th biased distribution is fixed and greater than zero.
\end{assumption}

We now prove that EMUS is consistent:
\begin{lemma}\label{lem: consistency}
Under Assumption~\ref{asm: defn of kappa} and the irreducibility condition~\eqref{eq: irreducibility condition}, $\pi_\emus[g]$ converges almost surely to $\pi[g]$ as the total number of samples $N$ tends to infinity.
\end{lemma}
\begin{proof}
Since the processes $X^i_t$ are ergodic,
\begin{equation}\label{eq: convergence of F and gi}
\bar F \convas F, \text{ } \bar g^\ast_i \convas g^\ast_i, \text{ and } \bar \1^\ast_i \convas \1^\ast_i \text{ as } N \rightarrow \infty.
\end{equation}
Moreover, by Lemma~\ref{lem: differentiability of z} in Appendix~\ref{apx: proof of clt}, $w(G)$ is continuous at $F$.
(Technically, $w(G)$ admits an extension to the set of all $L \times L$ matrices, which is continuous at $F$.)
Therefore, as a function of $\bar F$, $\bar g^\ast_i$, and $\bar \1^\ast_i$, $\pi_\emus[g]$ is continuous at $F$, $\pi_i[g^\ast]$, and $\pi_i [\1^\ast]$.
It follows by the continuous mapping theorem and equation~\eqref{eq: final decomposition formula} that ${\pi_\emus[g] \convas \pi[g]}$.
\end{proof}

\subsection{Iterative EMUS and Vardi's Estimator for Selection Bias Models}
\label{sec: iterative EMUS}

EMUS resembles certain methods for computing normalization constants of families of probability densities \cite{Geyer1994,Vardi:EmpDistSelectionBias1985,meng1996simulating,Kong2003}.
The resemblance arises because the weights in the third step of EMUS are the normalization constants of the biased distributions.
These methods have been used, for example, to compute Bayes factors in model selection problems~\cite{Geyer1994} and for computations related to selection bias models~\cite{Vardi:EmpDistSelectionBias1985}.
In this section, we explain how EMUS relates to Vardi's estimator for selection bias models~\cite{Vardi:EmpDistSelectionBias1985} and its descendants such as the popular Multistate Bennett Acceptance Ratio (MBAR) method~\cite{shirts2008statistically}.
In addition to comparing EMUS with these methods, we review a method, iterative EMUS, for solving the nonlinear system of equations defining Vardi's estimator~\cite{Thiede2016}.
The first iterate of this method is exactly the EMUS estimator described in Section \ref{sec:emus}. We note that our analysis and calculations in Sections~\ref{sec: error analysis},~\ref{sec: limiting results}, and~\ref{sec: numerical results} pertain only to the EMUS method and not iterative EMUS or MBAR. However, the similarity between the methods suggests that our results may generalize.


Given the notation developed in Section \ref{sec:emus}, we can express Vardi's estimate $\zv$ of the weight vector as the vector with entries $\zv_i = u_i\,N_i$ where $u$ solves the equation
\begin{equation}\label{eq: vardi estimating equation}
N_j = \sum_{i=1}^L  N_i\,  \bar F_{ij}(u),  \quad\sum_{i=1}^L u_i\,N_i = 1
\end{equation}
and where we have now made the dependence of the matrix $\bar F$ in \eqref{eq:barg1F} 
on the choice of  $u \in \Real^L$ explicit.
By \cite[Theorem~1]{Vardi:EmpDistSelectionBias1985}, this nonlinear equation determines  $z^V$ uniquely  whenever the irreducibility criterion of Lemma~\ref{lem: irreducibility condition} holds.

Vardi's estimator was originally derived assuming that the samples $X^i_t$ from the biased distributions were i.i.d. 
In that case, it is the nonparametric maximum likelihood estimator of the target distribution $\pi$ given samples from the biased distributions $\pi_i$~\cite{Vardi:EmpDistSelectionBias1985}, and it has certain optimality properties~\cite{gill1988}.
Several adjustments to the estimator have been proposed for the case of samples from Markov processes.
In the Multistate Bennett Acceptance Ratio (MBAR) method, one replaces the factors $N_i$ appearing in the summand in~\eqref{eq:barg1F} with \emph{effective sample sizes} $n_i$, which are computed from estimates of the integrated autocovariance of a family of functions~\cite{shirts2008statistically}.
(In some versions of MBAR, the sample average over all $N_i$ points is replaced with a sample average over the $n_i$ points obtained by including only every $N_i /n_i$'th point along the trajectory $X^i_t$.)
Another recent work proposes different effective sample sizes computed by minimizing an estimate of the asymptotic variance of the estimator~\cite{doss2014estimates}.
In fact, the estimator is consistent with $N_i$ replaced by any fixed positive number~\cite{doss2014estimates}.
We have found that our numerical results do not depend sensitively on the choice of effective sample size, so we use $N_i$ for simplicity.

We now review iterative EMUS, which we introduced in~\cite{Thiede2016}.
Iterative EMUS may be understood as a fixed point iteration for solving equation~\eqref{eq: vardi estimating equation}.
The iteration proceeds as follows:
\begin{enumerate}
 \item As an initial guess for $z^V$, choose a positive vector $z^0 \in \Real^L$. 
 Set $m=0$. Choose a tolerance $\tau >0$.
 \item Compute $\bar F_{ij}(u)$ for $u_i = z^m_i  / N_i$. Solve the eigenvector equation 
 \begin{equation}\label{eq: iterative emus eigenvector equation}
  w_j= \sum_{i=1}^L w_i\, \bar F_{ij}(u),\quad \sum_{i=1}^L w_i = 1
 \end{equation} 
 for $w \in \mathbb{R}^L$ 
to obtain the  updated estimate $z^{m+1}$ of $\zv$ with entries 
\[
z^{m+1}_i = \frac{u_i\, w_i}{\sum_{k=1}^L u_k w_k}.
\]
 \item If $\max_{i} \left\lvert w_i-\frac{N_i}{N}\right\rvert >\tau$, then increment $m$ and repeat step~2.
\end{enumerate}

In \cite{Thiede2016}  we show that the eigenvector equation~\eqref{eq: iterative emus eigenvector equation} has a unique solution for every $m$, and we suggest a numerical method for finding the solution.
We also discuss the convergence of iterative EMUS, and we show that for every fixed $m$, $z^m$ is a consistent estimator of the weight vector $z$.
If one chooses $z^0_i = N_i/N$, then $z^1$ is the EMUS estimate of $z$, $w(\bar F)$.

\subsection{Related MCMC methods}
EMUS belongs to a large class of MCMC methods that by various mechanisms promote a more uniform sampling of space. 
For example, in parallel tempering \cite{swendsen1986replica,Geyer1991paralleltemp}, one uses MCMC samples drawn from a distribution or sequence of distributions close to the uniform distribution to speed sampling of the target distribution.
The bias introduced by the choice of distributions is corrected either by reweighting the samples or by a replica exchange strategy~\cite{Geyer1991paralleltemp}.
The Wang--Landau~\cite{wang2001efficient} and Metadynamics methods~\cite{laio2002escaping} adaptively construct a biased distribution to achieve uniform sampling in certain coordinates. 
The temperature accelerated molecular dynamics method~\cite{maragliano2006TAMD} is also designed to achieve uniform sampling in a given coordinate, but it works by entirely different means.
In EMUS and other stratified MCMC methods, one achieves more uniform sampling by ensuring that each stratum contains points from at least one MCMC simulation.  EMUS is perhaps most similar in spirit to the parallel Markov chain Monte Carlo method~\cite{vanderwerken2013parallel}.

\section{Error Analysis of EMUS}
\label{sec: error analysis}
Here, we develop tools for analyzing the error of EMUS. 
First, in Section~\ref{sec: sensitivity of EMUS}, we prove a CLT for EMUS, and we derive a convenient upper bound on the asymptotic variance.
Then, in Section~\ref{sec: quantifying sampling error}, we analyze the dependence of the asymptotic variance of EMUS on the choice of biased distributions. 
We use these tools in Section~\ref{sec: limiting results} to prove limiting results demonstrating the advantages of EMUS for multimodal distributions and tail probabilities. In addition, our CLT for EMUS is the basis for both practical error estimates (Section~\ref{sec: numerical results} and Appendix~\ref{sec:error-bars}) and also a method of optimizing the allocation $\kappa_i$ of the samples among the different biased distributions~\cite{Thiede2016}.

\subsection{A CLT for EMUS and an Estimate of the Asymptotic Variance} \label{sec: sensitivity of EMUS}  
In this section, we prove a Central Limit Theorem (CLT) for EMUS, and we derive an upper bound on the asymptotic variance $\sigma^2_\emus(g)$ of $\pi_\emus[g]$.
To prove the CLT for EMUS, we must assume that a CLT holds for trajectory averages over the biased processes:
\begin{assumption}\label{asm: CLT for components}
For any matrix $H$, let  $H_{i:}$ denote the $i$'th row of $H$.
Define $\bar G \in \Real^{L \times 2}$ by 
\begin{equation*}
 \bar G_{i:} = (\bar{g}^\ast_i, \bar \1^\ast_i).
\end{equation*}
Assume that
\begin{equation}\label{eq: CLT for ith components}
  \sqrt{N_i} \left ( \left ( \bar{F}_{i:}, \bar{G}_{i:} \right ) - (F_{i:}, \pi_i[g^\ast], \pi_i[\1^\ast] ) \right )
  \convdist \N(0, \Sigma_i)
\end{equation}
for some asymptotic covariance matrix $\Sigma_i \in \Real^{(L+2) \times (L+2)}$ of the form
\begin{equation}\label{eq: form of ith covariance matrix}
  \Sigma_i 
  \ignore{=\begin{pmatrix}
  \acov \left (\bar{F}_{i:}\right ) & \acov \left (\bar{F}_{i:}, \bar G_{i:} \right)  \\
  \acov \left ( \bar G_{i:}, \bar{F}_{i:} \right ) & \acov \left (\bar G_{i:} \right )
  \end{pmatrix}}
  = \begin{pmatrix}
    \sigma^i & \rho_i \\
    \rho_i^\t & \tau_i
  \end{pmatrix},
\end{equation}
where $\sigma^i \in \Real^{L \times L}$ denotes the asymptotic covariance of $\bar{F}_{i:}$ with itself, $\rho_i \in \Real^{L \times 2}$ denotes the asymptotic covariance of $\bar{F}_{i:}$ with $\bar{G}_{i:}$, and $\tau_i \in \Real^{2 \times 2}$ denotes the asymptotic covariance of $\bar{G}_{i:}$ with itself. 
\end{assumption}

We expect a CLT to hold for most MCMC methods, target distributions, and target functions of interest in statistics and statistical mechanics. 
We refer to~\cite{roberts1997geometric} for a comprehensive review of conditions guaranteeing a CLT. 
In Theorem~\ref{thm: ergodicity and estimates of iat} of Section~\ref{sec: quantifying sampling error}, we prove a CLT and an estimate of the asymptotic variance for a simple family of processes which one might use to sample the biased distributions in an application of EMUS.

We now prove a CLT for $\pi_\emus[g]$, and we give a formula expressing the asymptotic variance $\sigma^2_\emus(g)$ of $\pi_\emus[g]$ in terms of the asymptotic variances $\Sigma_i$ of the trajectory averages. 
In this formula, $(I-F)^\#$ denotes the group generalized inverse of $I-F$; the group inverse $A^\#$ of a matrix $A$ is characterized by the properties  
\begin{align*}
A A^\# A = A \text{, } A^\# A A^\# = A^\# \text{, and } A A^\# = A^\# A.
\end{align*}
We refer to~\cite{GolMey:ComputingInvDist} for a detailed explanation of the properties of the group inverse, a proof that $(I-F)^\#$ exists whenever $F$ is stochastic and irreducible, and an algorithm for computing $(I-F)^\#$.

In Theorem~\ref{thm: CLT for EMUS} and below we impose the following assumption:
\begin{assumption}\label{asm: independent processes}
  The processes $X^i_t$ sampling the biased distributions are independent.
\end{assumption}

This assumption does not hold for all stratified MCMC methods. For example, in replica exchange umbrella sampling one periodically allows configuration exchanges between neighboring processes; see~\cite{liu2001monte} for a general discussion of replica exchange strategies and~\cite{sugita_multidimensional_2000} for an application of replica exchange in a method similar to EMUS. The result is a single process taking values in $\Real^{L \times d}$ and sampling the product distribution $\Pi(x_1,x_2, \dots, x_L) = \pi_1(x_1)\pi_2(x_2) \dots \pi_1(x_L)$. In this case, a CLT would still hold for EMUS, but the asymptotic variance would take a different form.  We also assume that $u_i=1$ for all $i$.  As already mentioned, we can equivalently absorb the choice of $u_i$ into the choice of the bias function $\psi_i$.

\begin{theorem}\label{thm: CLT for EMUS}
Let Assumptions~\ref{asm: defn of kappa},~\ref{asm: CLT for components}, and~\ref{asm: independent processes} and the irreducibility condition~\eqref{eq: irreducibility condition} hold. 
Let $g$ be square integrable over $\pi$, so $\pi[g^2] < \infty$.
Defining
\begin{equation*}
\Psi = \frac{1}{\sum_{i=1}^L \pi[\psi_i]}\quad\text{and}\quad
 \ell := \Psi 
 \left ( 1, -\pi[g] \right )^\t \in \Real^2.
\end{equation*}
Let $\mathfrak{g} \in \Real^L$ be the vector with $\mathfrak{g}_i := \ell \cdot (\pi_i[g^\ast], \pi_i[\1^\ast])$.
We have
\begin{equation}
  \sqrt{N} \left( \pi_\emus[g]  - \pi[g]\right )
  \convdist \N(0, \sigma^2_\emus(g)),
\end{equation}
where
\begin{align}\label{eq: formula for asymptotic variance}
  \sigma^2_\emus (g) = \sum_{i = 1}^L \frac{z_i^2}{\kappa_i} \big \{
  &(I-F)^\# \mathfrak{g}  \cdot \sigma^i (I-F)^\# \mathfrak{g} 
  + 2 (I-F)^\# \mathfrak{g}  \cdot \rho_i \ell + \ell^\t \tau_i \ell \big \}.
\end{align}
\ignore{(We note that $\rho_i \ell$ and $\ell^\t \tau_i \ell$ are asymptotic covariances of averages over the biased processes; we have 
\begin{equation*}
 \rho_i \ell = \acov(\bar{F}_{i:}, \bar h_i) \text{ and }  
 \ell^\t \tau_i \ell = \acov(\bar h_i),
\end{equation*}
where $h = \ell \cdot (g^\ast, \1^\ast)$, $\acov(\bar{F}_{i:}, \bar h_i)$ is the asymptotic variance of $\bar{F}_{i:}$ with $\bar h_i$, and $\acov(\bar h_i)$ is the asymptotic variance of $\bar h_i$.)
}
\end{theorem}

\begin{proof}
 The result follows using the delta method and a formula expressing $w'(F)$ in terms of $(I-F)^\#$; we give the details in Appendix~\ref{apx: proof of clt}.
\end{proof}

We now derive a convenient upper bound on the asymptotic variance $\sigma^2_\emus(g)$.
In Section~\ref{sec: limiting results}, we use this bound to analyze the efficiency of EMUS in the low-temperature limit and in the limit of small tail probabilities. 
Our bound is based on the probability $\P_i[t_j < t_i]$ defined below:
\begin{definition}\label{def: definitions of pij}
Let $Y_n$ be the Markov chain with state space $\{1,2, \dots, L\}$ and transition matrix $F$.
Let ${\P_i[t_j < t_i]}$ denote the probability that $Y_n$ hits $j$ before returning to $i$, conditioned on $Y_0 = i$.
\end{definition}

\begin{theorem}\label{thm: upper bound on asymptotic variance}
Let Assumptions~\ref{asm: defn of kappa},~\ref{asm: CLT for components}, and~\ref{asm: independent processes} and the irreducibility condition~\eqref{eq: irreducibility condition} hold.
  Let $g$ be square integrable over $\pi$, so $\pi[g^2] < \infty$. Let $\sigma^2(g)$ be the asymptotic variance of $\pi_\emus[g]$, and for any measure $\nu$ and function $f$ let $\var_\nu(f)$ be the variance of $f$ over $\nu$.
  Define the function
  \begin{equation*}
    h = g^\ast - \pi[g] \1^\ast,
  \end{equation*}
  and let 
  $
  \acov(\bar{h}_i) 
  $
  be the asymptotic variance of the trajectory average of $h$ over the biased process $X^i_t$.
  We have
  \begin{align}
    \sigma^2_\emus(g) 
    \leq 
    2\sum_{i=1}^L \frac{1}{\kappa_i}
    \Bigg \{
    &z_i^2 \Psi^2 \acov(\bar{h}_i) 
        + \tr(R^i)\pi[\lvert h \rvert ]^2  
      \sum_{\substack{j \neq i \\ F_{ij} >0}}
    \frac{\var_{\pi_i}(\psi_j^\ast)}{\P_{i}[t_j < t_i]^2}
    \Bigg \},
    \label{eq:upper bound on asymptotic variance}
  \end{align}
  where $R^i \in \Real^{L \times L}$ with 
  \begin{equation*}
    R^i_{jk}:= \frac{\sigma^i_{jk}}{\sqrt{\var_{\pi_i}(\psi_j^\ast)}\sqrt{\var_{\pi_i}(\psi_k^\ast)}}.
  \end{equation*}
\end{theorem}

\begin{proof}
 The result follows from Theorem~\ref{thm: CLT for EMUS}, using the perturbation bounds which we derived in~\cite{ThVKWe:Perturbation}. Details appear in Appendix~\ref{apx: proof of clt}.
\end{proof}

\subsection{Dependence of the Asymptotic Variance on the Choice of Strata} \label{sec: quantifying sampling error}

In this section, we consider how the choice of strata influences the factors in the upper bound~\eqref{eq:upper bound on asymptotic variance} on $\sigma^2_\emus(g)$. Roughly, the asymptotic variances $C(\bar h_i)$ and $\tr(R^i)$ characterize the sampling error, and for each $i$ the factor
  \begin{equation}\label{eq: key factor in section on relative variances}
   \sum_{\substack{j \neq i \\ F_{ij} >0}} \frac{\var_{\pi_i}(\psi_j^\ast)}{\P_{i}[t_j < t_i]^2}
\end{equation}
  measures the sensitivity of the EMUS estimator to sampling errors associated with $\pi_i$.

  We show in Section~\ref{subsec: rate of convergence} that the factors $C(\bar h_i)$ and $\tr(R^i)$ characterizing the sampling error may be controlled by decreasing the diameters of the strata.
  We show in Section~\ref{subsec: size of entries} that that $\P_{i}[t_j < t_i]$ may be controlled by ensuring sufficient overlap between neighboring strata. This last observation leads to a practical condition guiding the choice of strata; see Remark~\ref{rem: remark on size of entries in F and choice of strata} and~\eqref{eqn: size of sub and superdiagonal entries in calculations}.

Our theorems in this section apply to a specific class of strata and Markov processes that are broadly representative of those employed in practical applications. Thus, the assumptions made here are much stronger than those made in proving the CLT for EMUS, for example. We discuss how our results might extend to more general implementations of stratified MCMC after the statement of Theorem~\ref{thm: ergodicity and estimates of iat} in Section~\ref{subsec: rate of convergence} 
 and in Remark~\ref{rem: remark on size of entries in F and choice of strata}.
  
  \subsubsection{Asymptotic Variances of MCMC Averages}\label{subsec: rate of convergence}
 Here, we consider the effect of the choice of strata on the asymptotic variances $C(\bar h_i)$ and $\tr(R^i)$.
Because such a diverse variety of biased processes and distributions could in principle be used, it is futile in our opinion to try for a completely general result. Instead, motivated by the efficiency analysis undertaken in Section~\ref{sec: limiting results}, we introduce a simple parametric family of bias functions, and for this family we state Assumption~\ref{asm: dependence of iat on strata} relating the diameters of the strata with the asymptotic variances.  In Theorem~\ref{thm: ergodicity and estimates of iat}, we verify Assumption~\ref{asm: dependence of iat on strata} for one representative class of biased processes. Finally, at the end of this section, we explain why we expect the assumption to hold for other choices of biased processes and distributions.

Consider the following representative class of bias functions:
Given a family of sets ${\{U_i : i = 1, \dots, L\}}$ with ${\cup_{i=1}^L U_i = \Omega}$, define  
\begin{equation}\label{eq: pw constant bias functions}
 \psi_i := \1_{U_i} \text{ and } \pi_i (dx) := \frac{\1_{U_i}(x) \pi(dx)}{\pi[\1_{U_i}]} \text{ for } i = 1, \dots, L,
\end{equation}
where $\1_{U_i}$ denotes the characteristic function of $U_i$.
Assume that the sets $U_i$ are chosen so that the irreducibility criterion of Lemma~\ref{lem: irreducibility condition} holds. 
For example, suppose that $\Omega = [0,1]^d$ is the $d$-dimensional unit cube.
One might choose $K \in \mathbb{N}$, set $h := 1/K$, and define 
\begin{equation}\label{eq: uniform grid of strata}
 U_\iidx = (h[-1,1]^d + h\iidx) \cap \Omega \text{ for } \iidx \in \{0, 1, \dots, K\}^d,
\end{equation}
covering $\Omega$ uniformly by a grid of strata having diameters proportional to $h$.
We use this uniform grid as a device when analyzing the effect of the stratum size on the efficiency of EMUS in Section~\ref{sec: limiting results}.
However, while such a na\"ive choice may suffice for small $d$, it is not practical for large $d$.
We discuss appropriate bias functions for high-dimensional problems later in this section and again in Section~\ref{subsec: EMUS marginals}.

Since we wish to study grids like~\eqref{eq: uniform grid of strata} as $h=1/K$ varies, we state our assumption on asymptotic variances in terms of the following parametric family of strata: Let $x_0 \in \Omega$, and let $Z \subset \Real^d$ be a bounded set containing $0$.
For each $h > 0$, define a stratum and a biased distribution by
\begin{equation}\label{eq: Z form of bias functions}
  Z_h = x_0 + hZ 
  \text{ and } \pi_h(dx) = \frac{ \1_{Z_h} (x) \pi(dx) }{\pi[\1_{Z_h}]}.
\end{equation}
To make the connection with~\eqref{eq: uniform grid of strata}, one should imagine that $Z=[-1,1]^d$ and that $x_0$ is one of the grid points $h \mathbf{i}$.

Assumption~\ref{asm: dependence of iat on strata} characterizes the dependence of the asymptotic variance of MCMC averages over $\pi_h$ on the parameter $h$:

\begin{assumption}\label{asm: dependence of iat on strata}
 Assume that $f:\Omega \rightarrow \Real$ has finite variance $\var_{h}(f)$ over $\pi_h$, and define $\sigma^2_h(f)$ to be the asymptotic variance of an MCMC trajectory average approximating $\pi_h[f]$. Write
\begin{equation*}
 \pi(x) = \frac{\exp(-\beta V(x))}{\int \exp(-\beta V(y)) \, dy},
\end{equation*}
for some \emph{potential} $V : \Omega \rightarrow \Real$ and \emph{inverse temperature} $\beta >0$. 
  We assume
  \begin{align*}
    \frac{\sigma^2_h(f)}{\var_{\pi_h}(f)} &\leq C h^a \beta^b \exp \left (\beta \left ( \max_{Z_h} V - \min_{Z_h} V \right ) \right ) 
        \leq  C h^a \beta^b \exp \left (\beta h \diam(Z) \lVert \nabla V \rVert_\infty \right ) 
  \end{align*}
  for some $C,a,b \geq 0$ independent of $h$,$Z$, and $f$. 
\end{assumption}

To motivate Assumption~\ref{asm: dependence of iat on strata}, we prove that a special case holds for a representative class of processes sampling the biased distributions, cf.\@ Theorem~\ref{thm: ergodicity and estimates of iat}.
Assume that the potential $V$ appearing in the assumption is continuously differentiable.
Let $Z \subset \Real^d$ be either a convex polyhedron or a set with $C^3$ boundary.\footnote{The boundary of a set is $C^3$ if in a neighborhood of each point on the boundary, the boundary is the graph of a three times continuously differentiable function.}
Now let $X^h_t$ be the overdamped Langevin process with reflecting boundary conditions on $Z_h$.
This process is defined by the Fokker--Planck equation
\begin{alignat}{3}
  &\frac{\partial u}{\partial t}(x,t) = \diver (\beta^{-1} \nabla u(x,t) + u(x,t) \nabla V(x) ) 
  &&\text{ for } x \in U, t>0, \nonumber \\
  &(\beta^{-1} \nabla u(x,t) + u(x,t) \nabla V(x) )\cdot \mathbf{n}(x)=0 
  &&\text{ for } x \in \partial Z_h, t \geq 0, \text{ and }
  \label{eq: FKE for simple process}\\
  &u(x,0)=p(x) 
  &&\text{ for } x \in Z_h, \nonumber
\end{alignat}
where $\beta$ and $V$ are the inverse temperature and potential defined in Assumption~\ref{asm: dependence of iat on strata} and $\mathbf{n}(x)$ denotes the inward unit normal to $\partial Z_h$ at $x$.
That is, $X^h_t$ is the unique Markov process so that if $X^h_0$ has density $p(x)$, then $X^h_t$ has density $u(x,t)$.
The existence of the reflected process is established in~\cite{wang2013gradient,andres2009pathwise} when $Z$ is a convex polyhedron and in~\cite[Chapter~8]{ethier1986markov} when $Z$ has $C^3$ boundary.
A simple introduction to the reflected process and its properties appears in~\cite[Chapter~4]{pavliotis2014stochastic}.
We show in Theorem~\ref{thm: ergodicity and estimates of iat} that $X^{h}_t$ is ergodic for $\pi_h$, at least when $Z$ is bounded.   The reflected process $X^h_t$ shares many features with the processes used in practical stratified MCMC methods.  In particular, it is closely related to the (unreflected) overdamped Langevin process $Y_t$ \cite{RoDoFr78,RobTwe}.


\ignore{
However, there are some significant differences between the reflected process and practical methods: First, in computational chemistry, one typically chooses Gaussian bias functions instead of the piecewise constant bias functions of~\eqref{eq: pw constant bias functions}.
Second, one cannot compute trajectories of continuous time processes such as $X^h_t$, so one uses discretizations of these processes. 
Third, for high-dimensional problems, one typically stratifies only a certain low-dimensional \emph{reaction coordinate} or \emph{collective variable}.
We discuss these differences and their consequences in more detail after the statement of Theorem~\ref{thm: ergodicity and estimates of iat} below.
}

We now verify Assumption~\ref{asm: dependence of iat on strata} for the reflected process: 

\begin{theorem}\label{thm: ergodicity and estimates of iat}
 Assume that $f:\Omega \rightarrow \Real$ has finite variance $\var_{h}(f)$ over $\pi_h$. Let $Z$ either have $C^3$ boundary or be convex. Assume that $V$ is continuously differentiable.
 Suppose that $X^h_t$ is stationary; that is, $X^h_0$ has distribution $\pi_h$.
 Let $\bar f^h := \frac{1}{T} \int_{t=0}^T f(X^h_t) \, dt$ be the continuous time trajectory average of $f$.
 We have  
 \begin{equation*}
  \sqrt{T} (\bar f^h - \pi_h[f]) \convdist \N (0, \sigma^2_h(f)), 
 \end{equation*}
 where 
  \begin{equation}\label{eq: estimate of asymptotic variance for reflected diffusion}
    \sigma^2_h(f) \leq \Lambda h^2 \beta \exp \left (\beta \left ( \max_{Z_h} V - \min_{Z_h} V \right ) \right ) \var_{h}(f).
  \end{equation}
  The constant $\Lambda$ depends only on $Z$, not on $h$, $\beta$, $V$, or $f$.
\end{theorem}

\begin{proof} See Appendix~\ref{apx: estimate of integrated autocovariance}.
\end{proof}

\ignore{
We remark that~\eqref{eq: asymptotic variance in terms of generator} is a version of the usual formula equating the asymptotic variance and the integrated autocovariance:
In fact, for  
\begin{equation*}
C(t) = \int_{Z_h} \E_x[g(X^h_t)- \pi_h[g]] (g(x) - \pi_h[g]) \pi_h(x) \, dx
\end{equation*}
the autocovariance of $X^h_t$, we have 
\begin{equation*}
-\pi_h [ (g-\pi_h[g]) L^{-1} (g-\pi_h[g])] = \int_{t=0}^\infty C(t) \, dt;
\end{equation*}
see~\cite{kipnis1986clt} for details. }


  

There are three major differences between the choice of strata and sampling scheme specified in this subsection and those typical of practical applications:  First, in molecular simulations, one typically chooses Gaussian bias functions instead of piecewise constant.
Second, practical methods must be discrete in time, e.g.,\@ one might use a discretization of the continuous time process $X^h_t$. 
Third, for high-dimensional problems, one typically stratifies only a certain low-dimensional \emph{reaction coordinate} or \emph{collective variable}.

In the first case, for Gaussian bias functions, a version of Theorem~\ref{thm: ergodicity and estimates of iat} holds with minor adjustments; we omit the exact statement and proof for simplicity.
In the second case, for discretizations of Langevin dynamics, the asymptotic variances of trajectory averages are closely related to the corresponding averages for the continuous time dynamics: In fact, under some conditions on the potential $V$,  
\begin{equation}\label{eq: discrete to continuous asymptotic variance}
 \lim_{\Delta t \rightarrow 0} \Delta t C_{\Delta t}(f) = C(f),
\end{equation}
where $C_{\Delta t}(f)$ is the asymptotic variance of the trajectory average of $f$ for the discretization with time step $\Delta t$ and $C(f)$ is the asymptotic variance for the continuous time process \@~\cite[Section~3.2]{LelievreStoltz:PDEinMD}. For other discrete time processes, we expect Assumption~\ref{asm: dependence of iat on strata} to hold with different exponents $a$ and $b$. For example, the affine invariance property of the affine invariant ensemble sampler~\cite{goodman2010ensemble} suggests $a=0$. 

The third case is subtle. When $d$ is large, one typically stratifies only in a function ${\theta: \Omega \subset \Real^d \rightarrow \Real^\ell}$ with $\ell$ much smaller than $d$. 
To be precise, one might choose a uniform grid of nonnegative functions $\eta_i:\Real^\ell \rightarrow \Real$ defined as in~\eqref{eq: uniform grid of strata}, but with supports covering $\theta(\Omega) \subset \Real^\ell$ instead of $\Omega \subset \Real^d$.
One would then define the bias functions
\begin{equation}\label{eq: stratifying in theta}
 \psi_i(x) := \eta_i(\theta(x)).
\end{equation}
(We make a similar choice in our calculations in Section~\ref{sec: numerical results}, cf.\@ the natural stratification~\eqref{eq: biased dist in projection case}.)
For a clever choice of $\theta$, these biased distribution may be much easier to sample than the target distribution. 
For example, suppose that the marginal $\pi_\theta$ of $\pi$ in $\theta$ were multimodal, but that the conditional distributions $\pi(\cdot \mid \theta = \theta_0)$ were unimodal or otherwise easy to sample for each fixed $\theta_0$.
In that case, for $h$ sufficiently small, each biased distribution would be unimodal, hence easy to sample. 
(Recall that $h$ sets the diameters of the strata for the grid of bias functions defined in~\eqref{eq: uniform grid of strata}, so $h$ small means that the diameter of the support of $\eta_i$ is small.)
In free energy calculations, it is often possible to choose such a $\theta$ based on intuition or scientific principles; see~\cite{Thiede2016,lelievre2010free} for discussion. When computing tails or marginals, the problem itself typically suggests a particular $\theta$; cf.\@ the \emph{natural stratification} in Section~\ref{subsec: EMUS marginals}.

The reader will notice that bias functions of the form~\eqref{eq: stratifying in theta} will typically have infinite support, rendering the bound in Assumption~\ref{asm: dependence of iat on strata} useless. In this case, one might hope for a similar bound with the potential function $V$ replaced by the free energy
\begin{equation*}
  F(\theta) := -\beta^{-1} \log( \pi_\theta (\theta) ),
\end{equation*}
where $\pi_\theta$ is the marginal density of $\pi$ in $\theta$.  Roughly, this replacement will be valid when, under the dynamics of the MCMC processes sampling $\pi$, the distribution of any variable (any function of the process) converges rapidly to its conditional distribution under $\pi$ given the current value of the $\theta$ variable. This will occur, for example, when the marginal in $\theta$ is multimodal or otherwise difficult to sample, but the conditional distributions are easy to sample. 
In general an effective choice of $\theta$ will be one for which conditional equilibration given $\theta$ occurs much more rapidly than the overall time to convergence of the process. More on the effective dynamics of low-dimensional variables can be found in \cite{pavliotis2008multiscale} or \cite{legoll2010conditional}.

\subsubsection{Controlling the Probabilities $\P_{i}[t_j < t_i]$}\label{subsec: size of entries}

Here, we examine the effect of the choice of strata on the factors in display \eqref{eq: key factor in section on relative variances}
that appear in our upper bound~\eqref{eq:upper bound on asymptotic variance} on $\sigma^2_\emus(g)$.

We begin with a lemma estimating ${\var_{\pi_i}(\psi_j^\ast)}/{\P_{i}[t_j < t_i]^2}$ in terms of $F_{ij}$:
\begin{lemma}\label{lem: estimates of sensitivity in terms of overlaps}
 We have $$\frac{\var_{\pi_i}(\psi_j^\ast)}{\P_{i}[t_j < t_i]^2} \leq \frac{1}{F_{ij}}.$$
\end{lemma}

\begin{proof}
 We have $$\P_i[t_k < t_i ] \geq \P[X_1=k \vert X_0=i] = F_{ik},$$
  where $X_t$ denotes the Markov chain with transition matrix $F$.
  Therefore, since $\psi_j^\ast(x) \in [0,1]$,
  \begin{align*}
    \frac{\var_{\pi_i}(\psi_j^\ast)}{\P_{i}[t_j < t_i]^2}
    &\leq \frac{\var_{\pi_i}(\psi_j^\ast)}{F_{ij}^2} 
      \leq \frac{\pi_i[(\psi_j^\ast)^2]}{F_{ij}^2} 
      \leq \frac{\pi_i[\psi_j^\ast]}{F_{ij}^2} 
      = \frac{1}{F_{ij}}.
  \end{align*}
\end{proof}

We now estimate the size of $F_{ij}$ for piecewise constant bias functions such as the uniform grid~\eqref{eq: uniform grid of strata}:

\begin{lemma}\label{lem: size of entries in F}
  Assume as in~\eqref{eq: pw constant bias functions} that the bias functions are piecewise constant, and write ${\pi(x) \propto \exp(-\beta V(x))}$. We have
  \begin{equation*}
   F_{ij} \geq  \frac{\lvert U_i \cap U_j \rvert}{\lvert U_i \rvert \left \lVert \sum_{k=1}^L \1_{U_k} \right \rVert_\infty}\exp \left (\beta \left ( \min_{U_i} V - \max_{U_i} V \right ) \right )
 \end{equation*}
 In particular, for the uniform grid of strata~\eqref{eq: uniform grid of strata}, we have
 \begin{align}
   F_{\iidx\jidx}
   &\geq \frac{1}{4^d} \exp \left (\beta \left ( \min_{U_\iidx} V - \max_{U_\iidx} V \right ) \right ) 
     \geq  \frac{1}{4^d} \exp \left (-2\beta h \sqrt{d} \lVert \nabla V \rVert_\infty \right ) 
   \label{eq: lower bound on F for uniform grid of strata}
\end{align}
 for any $\iidx, \jidx \in \Z^d$ so that $F_{\iidx \jidx} >0$.
\end{lemma}
\begin{proof}
  We have 
  \begin{align*}
    F_{ij} &= \pi_i[\psi_j^\ast] = \pi_i \left [\frac{\1_{U_j}}{\sum_{k=1}^L \1_{U_k}} \right ]
             \geq \frac{\pi[U_i \cap U_j]}{\pi[U_i]} \frac{1}{\left \lVert\sum_{k=1}^L \1_{U_k} \right \rVert_\infty} \\ 
           &\geq \frac{\lvert U_i \cap U_j \rvert}{\lvert U_i \rvert} \frac{1}{\left \lVert\sum_{k=1}^L \1_{U_k} \right \rVert_\infty} 
    \exp \left (\beta \left ( \min_{U_i} V - \max_{U_i} V \right ) \right ) ,  
  \end{align*}
  which proves the first claim made in the statement of the lemma.
  
  Now, for the uniform grid of strata~\eqref{eq: uniform grid of strata}, the minimum nonzero value of $\lvert U_\iidx \cap U_\jidx \rvert / \lvert U_\iidx \rvert$ is $1/2^d$, attained when $\jidx = (1,1,\dots,1)+ \iidx$.
  Moreover, except for a set of measure zero, each  $x\in\Real^d$ lies within $2^d$ strata, so $\left \lVert \sum_{k=1}^L \1_{U_k} \right \rVert_\infty=2^d$.
  Finally, we have
  \begin{equation*}
    \max_{U_\iidx} V(x) - \min_{U_\iidx} V(x) \leq \diam(U_\iidx) \lVert \nabla V \rVert_\infty = 2 \sqrt{d} h \lVert \nabla V \rVert_\infty,
  \end{equation*}
  and the result follows.
\end{proof}

\begin{remark}[A Condition to Guide the Choice of Strata]
  \label{rem: remark on size of entries in F and choice of strata}
  Lemmas~\ref{lem: estimates of sensitivity in terms of overlaps} and~\ref{lem: size of entries in F} suggest a practical constraint on the choice of strata: To ensure that the calculation of the weights is not too sensitive to sampling errors, it will suffice to choose strata so that nonzero entries of $F$ are not too small. We let this condition guide the choice of strata in Section~\ref{sec: numerical results}, cf.\@~\eqref{eqn: size of sub and superdiagonal entries in calculations}. However, the condition is only sufficient, not necessary. For example, consider a uniform grid of Gaussian bias functions similar to~\eqref{eq: uniform grid of strata}, but with Gaussian densities having
  mean $\mu = h[-1,1]^d + h\iidx$ and variance $\sigma^2=h^2$
replacing the characteristic functions $\1_{U_\iidx}$. In that case, even though $F$ will be dense and may have some extremely small nonzero entries, one can still control~\eqref{eq: key factor in section on relative variances} by decreasing $h$, under some conditions on $\pi$. We omit the exact statement and proof for simplicity.
\end{remark}

Despite the exponential dependence on $d$ in~\eqref{eq: lower bound on F for uniform grid of strata}, EMUS and other stratified MCMC methods are advantageous for high-dimensional problems because it often suffices to stratify only a low-dimensional collective variable.  In such cases, the dimension of the grid of strata is much smaller than dimension of the state space $\Omega$; see our discussion of collective variables in Section~\ref{subsec: rate of convergence} and our computations in Section~\ref{sec: numerical results}. It is important to keep this in mind when reading our results below.
Also, one may define a uniform grid of strata so that~\eqref{eq: key factor in section on relative variances} increases only as $d^2$ with dimension, not exponentially.  We construct such a grid in Appendix~\ref{apx: sparse grid}.

\ignore{
  When $\Omega = [0,\infty)$ and the target is a tail probability of the form $p_M = \pi([M,\infty))$, some strata must be unbounded.
 Consider, for example, the bias functions depicted in Figure~\ref{fig: small prob basis}. In that case, one cannot control the size of the entries in $F$ simply by decreasing the size of the strata, since $F_{K,K-1}$ tends to zero as $h$ decreases. The resolution of this difficulty is subtle; we present the details in Section~\ref{sec: small probability limit}, since they pertain only to the calculation of tail probabilities. 
}

\section{Limiting Results as a Rationale for EMUS}\label{sec: limiting results}

In this section, we analyze the efficiency of EMUS in two limits: 
First, we consider a \emph{low temperature limit}, where we write $\pi(x) \propto \exp(-\beta V(X))$ and let the inverse temperature $\beta$ increase, concentrating the target distribution at its modes and intensifying the effects of multimodality on the efficiency of MCMC sampling.
Second, we consider the estimation of increasingly small tail probabilities. 
Our goal in each case is to elucidate the advantages and disadvantages of EMUS for a broad class of problems, providing a rationale for the use of the method.
We hope that others will use the tools of Section~\ref{sec: error analysis} in similar fashion to develop their own novel applications of EMUS.

\subsection{Limit of Low Temperature}\label{sec: low temp}

Let the target distribution take the form 
\begin{equation*}
 \pi_\beta(x) = \frac{\exp(-\beta V(x))}{\int \exp(-\beta V(x)) \, dx}
\end{equation*}
for some \emph{potential} $V$ and \emph{inverse temperature} $\beta > 0$, as in Section~\ref{sec: quantifying sampling error}.
In this section, we analyze the efficiency of EMUS in the \emph{low temperature limit} as $\beta$ tends to infinity with $V$ fixed. 
We observe that $\pi_\beta$ concentrates at its modes (the minima of $V$) in this limit.
As a consequence, MCMC methods for sampling $\pi_\beta$ undergo transitions between modes only rarely, which makes direct MCMC sampling increasingly inefficient.
To be precise, we show that the asymptotic variance of a trajectory average of the overdamped Langevin dynamics increases exponentially with $\beta$ in the worst case. 
On the other hand, we show that the asymptotic variance of the EMUS estimate of the same average increases only polynomially.
Therefore, EMUS is dramatically more efficient than direct sampling in the low temperature limit. 

We consider the low temperature limit because it provides a convenient sequence of increasingly difficult to sample multimodal distributions: By analyzing EMUS in the low temperature limit, we hope to clarify its advantages for multimodal problems in general.
We have no other interest in low temperature.

We now examine the overdamped Langevin dynamics 
\begin{equation}\label{eq: overdamped dynamics, low temp section}
 dX_t^\beta = -\nabla V(X_t^\beta) dt + \sqrt{2\beta^{-1}} dB_t 
\end{equation}
in the low temperature limit. 
(The overdamped Langevin dynamics is ergodic for $\pi_\beta$ under certain conditions on $V$; see~\cite{roberts1996exponential} for example.)
For typical potentials $V$, the generator
\begin{equation*}
 \gen := -\beta^{-1} \Delta + \nabla V \cdot \nabla
\end{equation*}
of~\eqref{eq: overdamped dynamics, low temp section} has a spectral gap that shrinks exponentially with $\beta$; that is, for some $c >0$,
\begin{equation}\label{eq: spectral gap is exponentially small}
  -\exp(-c \beta) \leq \lambda_1 <0, 
\end{equation}
where $\lambda_1$ is the greatest nonzero eigenvalue of $\gen$. We refer to~\cite[Section~2.5]{LelievreStoltz:PDEinMD} for a review of results on the spectrum of $\gen$, and we refer to~\cite{HelfferKN2004metastability} for precise conditions on $V$ which guarantee~\eqref{eq: spectral gap is exponentially small}. 
Now let $v_1$ be an eigenfunction corresponding to $\lambda_1$ normalized so that $\pi_h[v_1^2] =1$.
By formula~\eqref{eq: resolvent formula for asymptotic variance}, the asymptotic variance $\sigma^2_\beta(v_1)$ of the trajectory average of $v_1$ satisfies
\begin{align*}
 \sigma^2_\beta(v_1) &= -\pi_\beta [ v_1 L^{-1} v_1 ] = -\lambda_1^{-1} \pi_\beta [ v_1^2] = -\lambda_1^{-1} \geq \exp(c\beta),
\end{align*}
indicating that the cost of estimating $\pi[v_1]$ by direct MCMC grows exponentially with $\beta$.

Having analyzed the overdamped Langevin dynamics, we now examine EMUS in the low temperature limit. 
For convenience, we assume that $\Omega$ is the unit cube $[0,1]^d \subset \Real^d$ with periodic boundary conditions; to be more precise, we let ${\Omega = \Real^d / \mathbb{Z}^d}$ be the set of all points in $\Real^d$ with $x$ and $y$ identified if and only if $x-y \in \mathbb{Z}^d$. 
Periodic boundary conditions are typical of problems in chemistry and computational statistical mechanics.
We do not see any difficulties in generalizing our results to other types of domains.

As $\beta$ increases, we must make the supports of the bias functions smaller.
We accomplish this by adjusting the parameter $h$ in a uniform grid of bias functions similar to those defined in~\eqref{eq: uniform grid of strata}.
To be precise, we fix $K \in \mathbb{N}$, set $h := 1/K$, and define
\begin{equation}\label{eq: simple family of bias functions}
 \psi_{\iidx} (x) := \frac{1}{2^d} \1_{[-1,1]^d} (K(x-h \iidx)) \text{ for } \iidx \in \{0, 1, \dots, K-1\}^d. 
\end{equation}
This family of $K^d$ bias functions is a partition of unity over $\Omega$, and the support of the $\iidx$'th bias function is 
\begin{equation*}
 U_\iidx := h[-1,1]^d + h\iidx. 
\end{equation*}
For convenience, we treat the index $\iidx$ as an element of $\mathbb{Z}^d/K\mathbb{Z}^d$; that is, we let $\iidx$ be periodic with period $K$ in each of its components, identifying $(0,i_2, \dots, i_d)$ with $(K,i_2,\dots,i_d)$, for example.
Figure~\ref{fig: low temp} illustrates such a family of bias functions, and it demonstrates the appropriate relationship between $\beta$ and $h$.

\begin{figure}[ht]
\begin{center}
\includegraphics[width=0.6\linewidth]{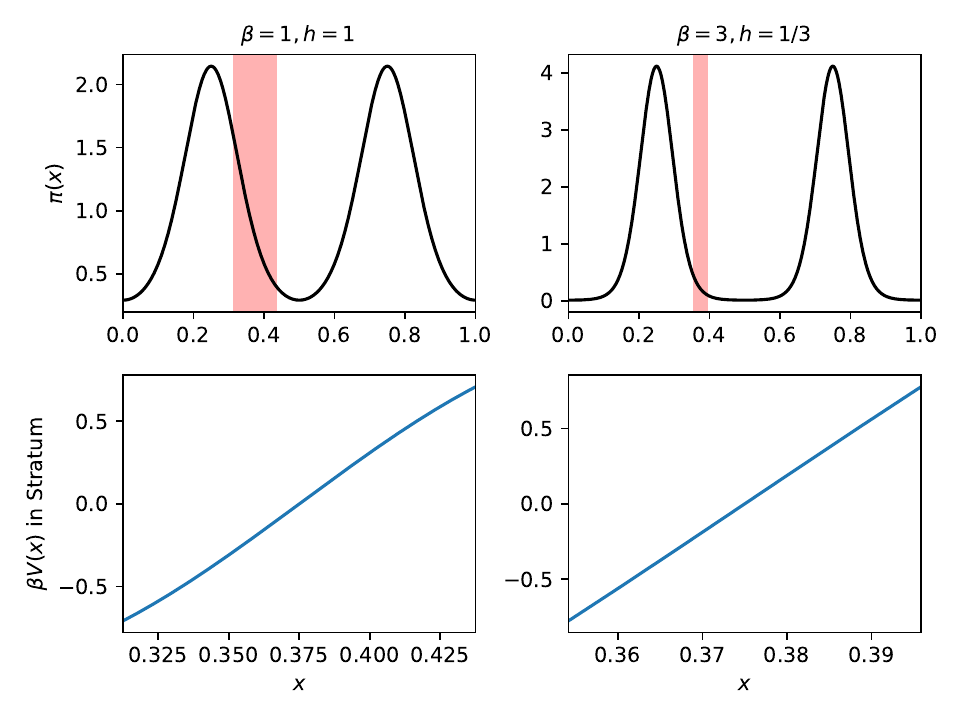}
\end{center}
 \caption{
 Bias functions and target distributions in the low temperature limit. In the upper two plots, the black curves are the densities of the target distributions for two different values of $\beta$. 
 Observe that $\pi$ concentrates at the minima of $V$ as $\beta$ increases.
 The red bands each lie above a single stratum chosen from a family of strata for which $h \propto \beta^{-1}$. 
 In the lower two plots, the blue curve is $\beta V(x)$ and the $x$-axis covers the bottom of the red band in the plot immediately above.
 Observe that the range of $\beta V(x)$ over the red band is the same for each of the two values of $\beta$.
 By Theorem~\ref{thm: ergodicity and estimates of iat} and the ensuing discussion in Section~\ref{sec: quantifying sampling error}, this implies that the cost of sampling a single biased distribution increases at most polynomially with $\beta$ when $h \propto \beta^{-1}$. 
 }
 \label{fig: low temp}
\end{figure}

We now show that the asymptotic variance of EMUS increases at most polynomially with $\beta$ when $K$ is chosen appropriately.
In light of the above discussion, this means that EMUS may be dramatically more efficient than direct sampling for multimodal problems.
We note that despite the exponential dependence on $d$ in~\eqref{eq: estimate for low temp limit} below, EMUS and other stratified MCMC methods are often advantageous for high-dimensional multimodal problems; see our discussion of low-dimensional collective variables in Section~\ref{sec: quantifying sampling error} and also our computations in Section~\ref{sec: numerical results}.

\begin{theorem}\label{thm: low temp EMUS}
For any bounded continuous function $g$, let $\sigma^2_{\beta,\emus} (g)$ denote the asymptotic variance of $\pi_{\beta,\emus}[g]$.
Let the bias functions be defined by~\eqref{eq: simple family of bias functions} with $K$ equal to the least integer greater than $\beta$; that is, 
\begin{equation*}
 K = \lceil \beta \rceil.
\end{equation*}
Take $\kappa_\iidx=1/K^d$. 
Let Assumption~\ref{asm: dependence of iat on strata} hold. 
We have 
\begin{equation}\label{eq: estimate for low temp limit}
\frac{\sigma^2_{\beta,\emus} (g)}{\var_{\pi_\beta} (g)} \leq  C (1+\beta)^{qd}
\end{equation}
for constants $C,q >0$ independent of $g$ and $\beta$, but depending on $V$ and the constants in Assumption~\ref{asm: dependence of iat on strata}.
\end{theorem}
\begin{proof}
  The proof is a straightforward application of the theory developed in Section~\ref{sec: error analysis}; we present the details in Appendix~\ref{apx: proof of low temp limit}.
\end{proof}

Our proof of Theorem~\ref{thm: low temp EMUS} relies on the perturbation bounds which we derived in~\cite{ThVKWe:Perturbation}.
These bounds allow one to estimate the sensitivity of $w(F)$ to small perturbations of $F$.
Most perturbation bounds in the literature predict that $w(F)$ is highly sensitive when the spectral gap of $F$ is small, but ours show that this is not always the case. 
(The spectral gap is $1 - \lvert \lambda_2 \rvert$, where $\lambda_2$ is the eigenvalue of $F$ with second largest absolute value.)
In the low-temperature limit, the spectral gap of $F$ decreases exponentially with $\beta$; see~\cite{ThVKWe:Perturbation} for a simple example of this phenomenon. 
Nonetheless, using our bounds, we show that the cost to compute averages by EMUS increases only polynomially in $\beta$.

\subsection{Limit of Small Probability}\label{sec: small probability limit}
In this section, we assess the performance of EMUS for computing tail probabilities.
To be precise, we let $\Omega=[0,\infty)$, and we consider estimation of probabilities of the form
\begin{equation*}
p_M := \pi([M, \infty)).
\end{equation*}
We show that for a broad class of distributions $\pi$, the cost of computing $p_M$ with relative precision by direct MCMC increases exponentially with $M$, whereas the cost by EMUS increases only polynomially.
Thus, EMUS is dramatically more efficient than direct sampling for computing the probabilities of tail events.

In Assumption~\ref{asm: assumptions for small prob limit} below, we state the conditions which we will impose on $\pi$ in our analysis.
These conditions specify a simple class of problems for which strong conclusions may be drawn. 
Similar results hold more generally. 
For example, in Section~\ref{sec: numerical results}, we report the results of a computational experiment demonstrating the advantages of EMUS for computing tails of a marginal density.

\begin{assumption}\label{asm: assumptions for small prob limit}
Write 
\begin{equation*}
\pi(x) = \exp(-V(x))
\end{equation*}
for some potential function $V: [0,\infty) \rightarrow \Real$. Assume that for some $M_0 \geq 0$:
\begin{enumerate}
\item Whenever $x \geq M_0$,
\begin{equation}\label{eq: convexity condition in small prob limit}
 0 \leq V''(x) \text{ and } 0 < V'(x).
\end{equation}
\item For some $\alpha \in (0,1)$ and $c>0$, whenever $x \geq M_0$,
\begin{equation}\label{eq: ergodicity condition in small prob limit}
 \alpha V'(x)^2 - V''(x) \geq c >0.
\end{equation}
\end{enumerate}
For example, we might have 
\begin{equation*}
 \pi(x) \propto \exp(-\lvert x \rvert^r) \text{ for any } r \geq 1.
\end{equation*}
\end{assumption}

  Condition~\eqref{eq: ergodicity condition in small prob limit} in Assumption~\ref{asm: assumptions for small prob limit} implies geometric ergodicity of the overdamped Langevin dynamics with potential $V$~\cite{roberts1997geometric}.
  We rely on this fact to motivate Assumption~\ref{asm: sampling unbounded strata} concerning the convergence of MCMC processes sampling biased distributions with unbounded support. 
Interestingly, we use the same condition to prove lower bounds on some of the entries of the overlap matrix; cf.\@ Lemma~\ref{lem: lower bound on F in small prob limit}.

Condition~\eqref{eq: convexity condition in small prob limit} in Assumption~\ref{asm: assumptions for small prob limit} implies  
\begin{equation*}
 p_M \leq D \exp(-\gamma M)
\end{equation*}
whenever $M \geq M_0$ for some $D, \gamma >0$. 
Therefore, the relative variance $\rho_M^2$ of $\1_{[M, \infty)}$ over $\pi$ satisfies 
\begin{equation*}
 \rho^2_M = \frac{p_M - p_M^2}{p_M^2} \geq D^{-1} \exp(\gamma M) -1.
\end{equation*}
We conclude that estimating $p_M$ with relative accuracy by a direct MCMC method (or even Monte Carlo with independent samples) requires a number of samples increasing exponentially with $M$.

By contrast, we show that for an appropriate choice of bias functions, the cost to estimate $p_M$ by EMUS increases only polynomially in $M$. 
For each $M>0$ and $K \in \mathbb{N}$, let
\begin{equation*}
 h:=\frac{M}{K},
\end{equation*}
and define the family of $K+2$ bias functions 
\begin{equation}\label{eq: bias functions for tails}
\psi_i(x) := 
\begin{cases}
 \frac12 \1_{[0,h]}(x) &\text{for } i=0, \\
 \frac12 \1_{[(i-1)h, (i+1) h]}(x) &\text{for } i=1,\dots,K-1 \\
 \frac12 \1_{[M-h,\infty)}(x) &\text{for } i=K, \text{ and } \\
 \frac12 \1_{[M,\infty)}(x) &\text{for } i = K+1.
\end{cases}
\end{equation}
As in Section~\ref{sec: low temp}, let $U_i$ denote the support of $\psi_i$.
This family of bias functions is a partition of unity on $[0,\infty)$; see Figure~\ref{fig: small prob basis}.

\begin{figure}[ht]
\begin{center}
\includegraphics[width=0.5\linewidth]{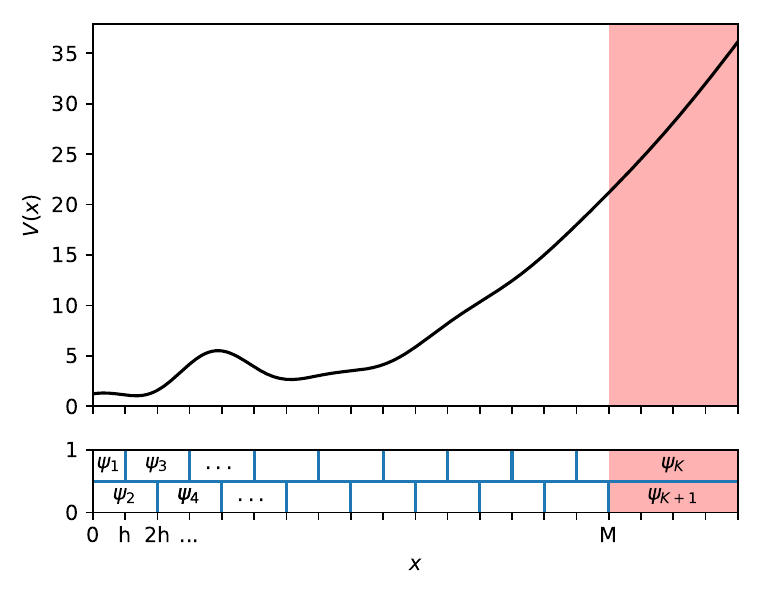}
\end{center}
\caption{The bias functions $\{\psi_i: i = 0, \dots, K+1\}$ defined in~\eqref{eq: bias functions for tails} and a potential function $V$ satisfying Assumption~\ref{asm: assumptions for small prob limit}.
Observe that the bias functions $\psi_{K}$ and $\psi_{K+1}$ have unbounded support. }
\label{fig: small prob basis}
\end{figure}

We now address the cost of estimating $p_M$ by EMUS.
First, we observe that Assumption~\ref{asm: dependence of iat on strata} on the asymptotic variances of MCMC averages does not cover the sampling of $\pi_K$ and $\pi_{K+1}$, since the supports of these distributions are unbounded. Thus, we require the following assumption.

\begin{assumption}\label{asm: sampling unbounded strata}
  Let $f: [0,\infty) \rightarrow \Real$, and define $\sigma^2_i(f)$ to be the asymptotic variance of an MCMC trajectory average approximating $\pi_i[f]$ for $i =K,K+1$. We assume
  \begin{equation*}
    \frac{\sigma^2_i(f)}{\var_{\pi_i}(f)} \leq D
  \end{equation*}
  for some $D$ independent of $M$ and $f$. 
\end{assumption}
In fact, since Assumption~\ref{asm: assumptions for small prob limit} implies that the overdamped Langevin dynamics is ergodic for $\pi(x) = \exp(-V(x))$ on the unbounded domain $\Omega = \Real$, we fully expect (but do not prove here) that under Assumption~\ref{asm: assumptions for small prob limit}, Assumption~\ref{asm: sampling unbounded strata} holds for overdamped Langevin constrained (by reflection as in~\eqref{eq: FKE for simple process}) to remain in the support of $\pi_K$ or $\pi_{K+1}.$  
Alternatively the reader  may simply assume that we draw i.i.d.\@ samples from the biased distributions.
All our results hold in that case.

\ignore{
In fact, considering the case of i.i.d.\@ samples reveals that the reasons for the advantages of EMUS are quite different in the low-temperature and small probability cases:
In the low-temperature limit, EMUS is advantageous because Markov processes sampling the biased distributions converge more quickly than processes sampling the target distribution. 
To be precise, the spectral gap of the generator for the unbiased process is much smaller than the spectral gap for the biased processes.
In the small probability case, EMUS is advantageous because the very small tail probability is represented as a function of much larger quantities, such as the entries in the overlap matrix.
}

\ignore{
    When $\Omega = [0,\infty)$ and the target is a tail probability of the form $p_M = \pi([M,\infty))$, some strata must be unbounded.
 Consider, for example, the bias functions depicted in Figure~\ref{fig: small prob basis}. In that case, one cannot control the size of the entries in $F$ simply by decreasing the size of the strata. For observe that as $h$ decreases in Figure~\ref{fig: small prob basis}, so does $F_{K,K-1}$. The resolution of this difficulty is subtle; we present the details in Section~\ref{sec: small probability limit}, since they pertain only to the calculation of tail probabilities. 
}

We show in Theorem~\ref{thm: small prob EMUS} that the relative asymptotic variance of the EMUS estimate of $p_M$ grows only polynomially with $M$ for a broad class of target distributions $\pi$. 
Therefore, EMUS may be dramatically more efficient than direct MCMC sampling when the goal is to compute tail probabilities. 
We observe that while the hypotheses of the theorem are somewhat restrictive, similar results hold more generally; for example, see Section~\ref{sec: numerical results} where we compute tails of a marginal density.

\begin{theorem}\label{thm: small prob EMUS}
Let Assumptions~\ref{asm: dependence of iat on strata},~\ref{asm: assumptions for small prob limit}, and~\ref{asm: sampling unbounded strata} hold.
Set 
\begin{equation*}
 K = M \max_{x \leq M} \lceil \lvert V'(x) \rvert \rceil.
\end{equation*}
Define a family of $K+2$ bias functions $\psi_i$ by~\eqref{eq: bias functions for tails}.
Take $\kappa_i=1/(K+2)$.
Let $\sigma^2_{M,\emus}$ denote the asymptotic variance of the EMUS estimate of $p_M$. 
We have
\begin{equation*}
 \frac{\sigma^2_{M,\emus}}{p_M^2} \leq C K^2
\end{equation*}
for some constant $C>0$ depending on $V$ but not on $M$.

For example, suppose that 
\begin{equation*}
 V(x) = \tilde V(x) + x^r,
\end{equation*}
where $\tilde V$ has bounded support and $r \geq 1$.
Then $\lvert V'(x) \rvert \leq C(1+M^{r-1})$, and so
\begin{equation*}
 \frac{\sigma^2_{M,\emus}}{p_M^2} \leq C M^2 (1+M^{r-1})^2.
\end{equation*}
\end{theorem}

\begin{proof}
  The proof is similar to the low temperature limit, Theorem~\ref{thm: low temp EMUS}, but with complications arising because not all strata are bounded and because here we consider the relative variance instead of the variance; see Appendix~\ref{apx: proof of low temp limit}. In particular, we require Assumption~\ref{asm: assumptions for small prob limit} to show that one can in fact choose $h$ so that all nonzero entries of $F$ are bounded above zero uniformly as $M$ increases; cf.\@ Lemma~\ref{lem: lower bound on F in small prob limit}. This is the only part of the proof relying on Assumption~\ref{asm: assumptions for small prob limit}.
\end{proof}

\section{EMUS for tails: An example from Bayesian inference}
\label{sec: numerical results}

We demonstrate the use of EMUS for efficiently exploring and visualizing distributions.
In particular, we show how EMUS may be used to efficiently compute both marginal densities and also tail probabilities of the form $\P[\eta(Z) \geq \eps^{-1}]$ where $\eta(Z)$ is a real valued function of a high-dimensional random variable $Z$.
For both tails and marginals, there is a natural and easy to implement choice of strata, which we describe in Section~\ref{subsec: EMUS marginals}. 

In Section~\ref{subsec: numerical experiments}, we calculate two different one-dimensional marginals of the posterior distribution of the hierarchical Bayesian mixture model described in Section~\ref{subsec: mixture model}. For one marginal, the natural stratification suffices. For the other, it does not, but a preliminary computation made with the natural stratification suggests a better choice of strata. We use this example to explain how to diagnose and correct problems related to poorly chosen strata: Our results will serve to guide the practice of stratified MCMC.   

\subsection{The natural stratification for tails and marginals}\label{subsec: EMUS marginals}

Here, we briefly explain how EMUS can be used to estimate tail probabilities and low-dimensional marginals of high-dimensional distributions. 
Let $\Omega \subset \Real^d$; let $\pi$ be a probability distribution on $\Omega$; and let $\eta : \Omega \rightarrow \Real$.
Suppose that one wishes to estimate the very small tail probability $\P[\eta(Z) \geq \eps^{-1}]$.
In this case, it is natural to stratify in $\eta$ only. 
That is, one may choose a partition of unity $\{\phi_i\}_{i=1}^L$ on $\Real$ and define bias functions 
\begin{equation}\label{eq: biased dist in projection case}
 \psi_i(x) = \phi_i(\eta(x)) \text{ for } i =1, \dots, L
\end{equation}
depending only on $\eta$.
For a partition of unity, one might choose the  regular grid of piecewise constant functions defined in Section~\ref{sec: small probability limit}.
We refer to \eqref{eq: biased dist in projection case} as the \emph{natural stratification}.
To compute the tail probability, one uses EMUS to estimate $\pi(\1_{[\eps^{-1}, \infty)} \circ \eta)$.

Computing marginal densities is similar; in fact, computing tails may be understood as a special case of computing a marginal density. Suppose now that $\eta:\Omega \rightarrow \Real^\ell$. To estimate the marginal $\pi_\eta$ of $\pi$ in $\eta$, one chooses a partition of unity $\{\phi_i\}_{i=1}^L$ on $\Real^\ell$, again defining bias functions by~\eqref{eq: biased dist in projection case}. One then uses EMUS to compute averages of \emph{histogram bins}, which are functions of the form 
\begin{equation}\label{eq: defn of histogram bin}
 b_{\eta_0}(\eta(x)) = \1_{\eta_0 + h[-1,1]^\ell}(\eta(x)).
\end{equation}
We have 
\begin{equation*}
 \lim_{h \rightarrow 0} \frac{1}{(2h)^\ell} \pi[b_{\eta_0}] = \pi_\eta(\eta_0),
\end{equation*}
so for small $h$ the averages of the histogram bins approximate~$\pi_\eta$.

By the argument in Section~\ref{sec: small probability limit}, EMUS with the natural stratification will be dramatically more efficient than direct sampling \emph{as long the biased distributions are no harder to sample than the target distribution $\pi$.} 
Essentially, this is because, with the natural stratification, very small averages like $\P[\eta(Z) \geq \eps^{-1}]$ over the target distribution $\pi$ are expressed as functions of much larger averages over the biased distributions $\pi_i$. Unfortunately, however, for general functions $\eta$, the biased distributions of the natural stratification need not be easy to sample. In Section~\ref{subsec: numerical experiments}, we give one example where the natural stratification works and one where it does not. In the case where it does not, we explain how to make a better choice of strata.

\subsection{A hierarchical Bayesian mixture model}\label{subsec: mixture model}
Here, we review the hierarchical Bayesian mixture model proposed in~\cite{richardson1997bayesmix}, and we discuss the difficulties which complicate inference under this model. As a tutorial in the use of EMUS, we present a numerical investigation of these difficulties in Section~\ref{sec: numerical results}.

In the hierarchical mixture model, the data vector 
$
\mathbf{y} = (y_1, \dots ,y_n) \in \Real^n
$
consists of independent identically distributed samples drawn from a mixture distribution of the form
\begin{equation*}
  p(y_i \vert \phi) = \sum_{k=1}^K q_k \nu(y_i; \mu_k, \lambda_k^{-1}), 
\end{equation*}
where $K$ is the number of mixture components, $q_k$ is the weight of the $k$'th mixture component, $\nu(\cdot; \mu_k, \lambda_k^{-1})$ is the normal density with mean $\mu_k$ and variance $\lambda_k^{-1}$, and $\phi$ is the vector of parameters
\begin{equation*}
 \phi = (\mu_1, \dots, \mu_K, \lambda_1, \dots, \lambda_K, q_1, \dots, q_{K-1}).
\end{equation*}
(Since $p(y_i \vert \phi)$ is a probability distribution, $q_1 + \dots + q_K=1$, and ${q_1, \dots, q_{K-1}}$ determine $q_K$.)
The following prior distribution is imposed on $\phi$:
\begin{align*}
 \mu_i &\sim \N (m, \kappa^{-1})\\
 \lambda_k &\sim \text{Gamma}(\alpha, \beta) \\
 \beta &\sim \text{Gamma}(g,h) \\
 (q_1, \dots, q_{K-1}) &\sim \text{Dirichlet}_K(1,\dots, 1).
\end{align*}
As in~\cite{jasra2005,Chopin2012}, we choose 
\begin{equation*}
 m=M, \text{ } \kappa=\frac{4}{R^2}, \text{ } \alpha=2, \text{ } g=0.2 , \text{ and } h=\frac{100g}{\alpha R^2}
\end{equation*}
 where $R$ and $M$ are the range and the mean of the observed data, respectively.
The posterior density is
\begin{align*}
 p(\theta \vert \mathbf{y}) &= 
 \frac{\kappa^{K/2}g^h \beta^{K\alpha +g-1}}{Z_K \Gamma(\alpha)^K \Gamma(g) (2\pi)^{\frac{n+K}{2}}}
 \left (\prod_{k=1}^K \lambda_k \right )^{\alpha-1} \\
 &\qquad \times \exp \left \{ -\frac{\kappa}2 \sum_{k=1}^K(\mu_k-M)^2 - \beta \left (h+\sum_{k=1}^K \lambda_k \right )  \right \} \\
 &\qquad \times \prod_{i=1}^N \left ( \sum_{k=1}^K q_k \lambda_k^{\frac12}
 \exp \left \{\frac{\lambda_k}2 (y_i - \mu_k)^2 \right \} \right ),
\end{align*}
where $\theta = (\phi,\beta)$ denotes the vector of all parameters to be inferred, including the hyperparameter $\beta$.

Several factors complicate inference based on this model. First, the mixture components are not identifiable; that is, the posterior distribution is invariant under permutation of the labels of the mixture components. 
Consequences of non-identifiability are discussed at length in~\cite{jasra2005,Chopin2012}. 
In our computations in Section~\ref{subsec: numerical experiments}, we impose the constraint 
\begin{equation*}
\mu_1 \leq \mu_2 \leq \dots \leq \mu_K 
\end{equation*}
to ensure that the components are identifiable.
Second, in Lemma~\ref{lem: posterior is unbounded}, we show that the posterior density may be unbounded, introducing spurious modes with infinite density. 
Finally, even with identifiability constraints, the posterior distribution may have multiple modes of finite posterior density. 
For example, see the modes reported in~\cite{Chopin2012}.
In Section~\ref{subsec: numerical experiments}, we use EMUS to efficiently visualize the posterior, assessing the effects of multimodality and unboundedness.

We suspect that the unboundedness of the posterior for this model is well known. However, we are unable to find a reference, so we now explain.
It is certainly well known that the likelihood of a Gaussian mixture model is unbounded: Roughly speaking, the likelihood is infinite when any mixture component is collapsed on a single data point~\cite{aitkin2001}. 
Nonetheless, one might expect the posterior density $p(\theta \vert \mathbf{y})$ to be bounded, since the prior penalizes large values of the precisions $\lambda_i$.
This is not always the case when the data vector contains repeated entries:

\begin{lemma}\label{lem: posterior is unbounded}
If any datum $y_i$ has frequency $N_i$ greater than 
\begin{equation*}
2g+2(K-1)\alpha,
\end{equation*}
then the posterior density $p(\theta \vert \mathbf{y})$ is unbounded.
\end{lemma}

\begin{proof}
Take the limit of $p(\theta \vert \mathbf{y})$ as $\lambda_1 \rightarrow \infty$ with $\mu_1=y_i$, $\beta = \lambda_1^{-1}$, and all other variables held fixed. 
\end{proof}

The reader will observe that under the model, the set of data vectors with repeated entries has probability zero. However, in practice, the data consist of measurements with finite precision, and therefore repeated entries occur commonly, cf.\@ the Hidalgo stamp data used in Section~\ref{subsec: numerical experiments}.

\subsection{Numerical experiments: Choosing strata, computing tails, diagnosis of problems}
\label{subsec: numerical experiments}
In this section, we explain how to recognize and correct problems related to poor choices of strata, and we demonstrate the use of EMUS to investigate the multimodality and unboundedness of the posterior in the mixture model.
We first compute two one-dimensional marginals of the high-dimensional posterior density $p(\theta \vert \mathbf{y})$ using the natural stratification~\eqref{eq: biased dist in projection case}.
The natural stratification works in one case but not the other. In the case where the natural stratification does not work, preliminary calculations based on the natural stratification suggest a better choice of strata. 

Here, we let $\mathbf{y}$ be the Hidalgo stamp data set first studied in~\cite{izenman1988philatelic}, consisting of the thicknesses of $485$ stamps, ranging between $60\ \mu \rm{m}$ and $130\ \mu \rm{m}$. 
We let there be three mixture components ($K=3$), following previous computational studies~\cite{Chopin2012,jasra2005}.
In our first calculation, we estimated the marginal in $\mu_2$ using the natural stratification with a grid of 201 bias functions covering the range $[7,11]$, with the support of the leftmost and rightmost bias functions reaching to $-\infty$ and $\infty$, respectively.
For the middle strata, define $\phi_1: \Real \rightarrow \Real$ by 
\begin{equation}
 \phi_1 (x) := \max \{0, 1 - \lvert x \rvert \}.
\end{equation}
We used the  bias functions
\begin{equation}\label{eq:bias-functions-mu2}
  \psi_{i}(\theta) =
  \phi_1 \left (\frac{\mu_2 - (7+(i-1)h)}{h} \right ), \text{ where } h:= 0.02
\end{equation}
for $i=2, \dots, 200$.
Now, define $\phi_2: \Real \rightarrow \Real$ by
\begin{equation}
    \phi_2 (x) := \min\{\max \{0, 1 - x \}, 1\}
\end{equation}
The first and last bias functions were
\begin{align}
    \psi_{1}(\theta) &=
  \phi_2 \left (\frac{\mu_2 - 7}{h} \right ) \\
    \psi_{201}(\theta) &=
    \phi_2 \left (\frac{(7 + 200 h)- \mu_2}{h} \right ),
\end{align}
where $h=0.02$ as before.

We chose the total number of bias functions based on the sizes of the off-diagonal entries in the overlap matrix.
For any bias functions of the form~\eqref{eq:bias-functions-mu2}, the overlap matrix is tridiagonal.
Thus, by Remark~\ref{rem: remark on size of entries in F and choice of strata}, if the superdiagonal and subdiagonal entries $F_{i,i+1}$ and $F_{i,i-1}$ are sufficiently large, then the EMUS estimator is not too sensitive to statistical errors in $\bar F$.
For our choice of bias functions,
\begin{equation}\label{eqn: size of sub and superdiagonal entries in calculations}
 \min\{F_{i,i+1}; i=1, \dots 200\} \geq 0.01 \text{ and } \min\{F_{i,i-1}; i=2, \dots, 201\} \geq 0.004. 
\end{equation}

We sampled the biased distributions using the affine invariant ensemble sampler with $100$ walkers, as implemented in the emcee package~\cite{foreman2013emcee}.
Due to computational restrictions on memory, only every tenth sample point was saved.
As a check on the sampling, the average acceptance probability over all walkers in the ensemble sampler was calculated for each biased distribution.  
Averaging over biased distributions gave a total average acceptance probability of 0.31.  
The minimum acceptance probability over all distributions was 0.12.

To initialize sampling, we computed an unbiased test trajectory; that is, a trajectory having ergodic distribution $\pi$.
We then started by sampling a single biased distribution $\pi_k$, initializing with points drawn randomly from the unbiased trajectory.
We sampled the other biased distributions in sequence, initializing with points drawn randomly from samples of adjacent biased distributions. 
Thus, we sampled $\pi_k$ first, then $\pi_{k-1}$ and $\pi_{k+1}$, then $\pi_{k-2}$ and $\pi_{k+2}$, etc.
We equilibrated the sampler in each $\pi_i$ for 3000 Monte Carlo steps, and collected data for an additional 100000 Monte Carlo steps.   Each step of the ensemble sampler involves perturbing the positions of each of the 100 walkers.

We computed the marginal in $\mu_2$ using a grid of 200 histogram bins, covering the region $[7,11]$; this corresponds to taking ${h=0.01}$ in~\eqref{eq: defn of histogram bin}. The result is the curve labeled EMUS in Figure~\ref{fig:mu2-marginal}.
The marginal in $\mu_2$ has two modes, labeled $1$ and $2$ in Figure~\ref{fig:mu2-marginal}.
We plot the mixture distributions corresponding to these modes in Figure~\ref{fig:mu2-mode-mixtures}.
(To be precise, the distributions in Figure~\ref{fig:mu2-mode-mixtures} correspond to means over histogram bins centered at the labeled points.)

For comparison, we also estimated the marginal in $\mu_2$ from multiple long, unbiased trajectories.  
We computed $100$ unbiased trajectories of the affine invariant ensemble sampler in parallel.
For each trajectory, the ensembles were first equilibrated for 10000 Monte Carlo steps, and then data were collected for 100000 steps.
These trajectories were combined and binned to produce the density labeled Unbiased in  Figure~\ref{fig:mu2-mode-mixtures}.
We estimated the relative asymptotic variance of the marginal density for the unbiased calculation using ACOR~\cite{acor},
and we estimated the relative asymptotic variance for the EMUS calculation using the method outlined in Appendix~\ref{sec:error-bars}.
We present the results in Figure~\ref{fig:mu2-marginal}.
Note that near the mode, unbiased MCMC performs slightly better than EMUS, but in the tails, EMUS performs dramatically better. 

\begin{figure}[ht]
  \subfloat[\label{fig:mu2-marginal}]{\includegraphics[width=0.5\linewidth]{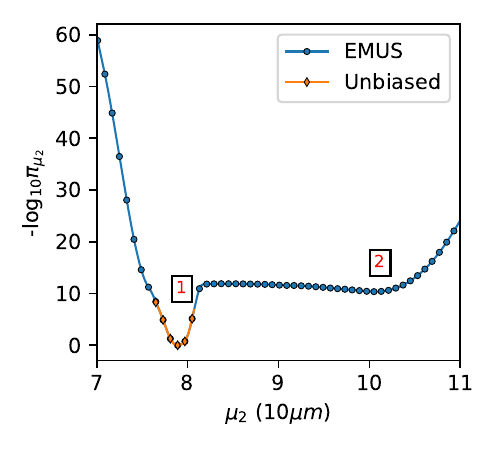}}
\subfloat[\label{fig:mu2-marginal_err}]{\includegraphics[width=0.5\linewidth]{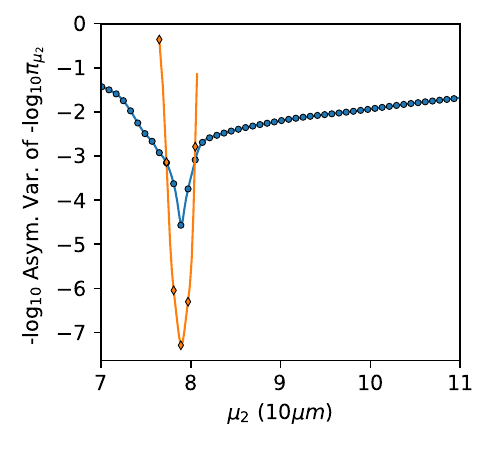}}
\caption{Estimates of the logarithm of the marginal density in $\mu_2$ and the asymptotic variances of those estimates. 
 Figure~\ref{fig:mu2-marginal} displays estimates of the marginal in $\mu_2$ computed by EMUS and by an unbiased trajectory of the ensemble sampler. 
Figure~\ref{fig:mu2-marginal_err} displays the asymptotic variances of these two estimates of the marginal density. We note that while the unbiased calculation has greater accuracy near the mode, the EMUS calculation has greater accuracy in the tails.
The relative errors in this figure were estimated using the method described in Appendix~\ref{sec:error-bars}. }
\end{figure}

\begin{figure}[ht]
  \begin{center}
    \includegraphics[width=\linewidth]{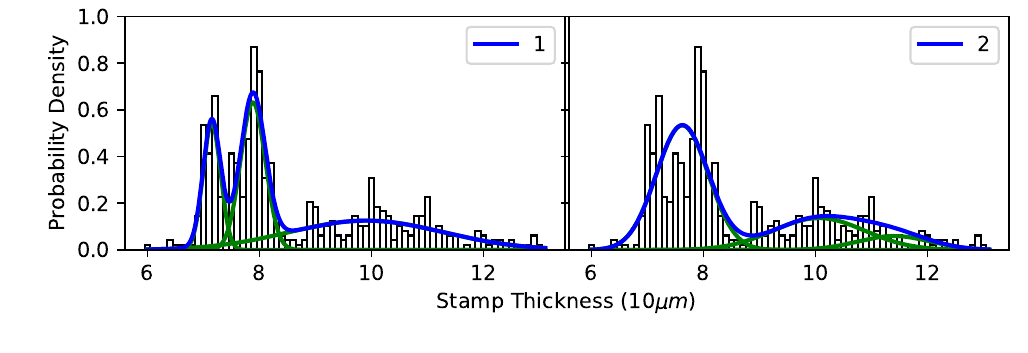}
  \end{center}
  \caption{Gaussian mixtures corresponding to modes of the marginal in $\mu_2$. 
    Mixtures $1$ and $2$ correspond to the labeled points in Figure~\ref{fig:mu2-marginal}.
    To be precise, the blue curve in each plot is the mixture distribution corresponding to the mean of a histogram bin centered at the point labeled in Figure~\ref{fig:mu2-marginal}.
    The green curves are the individual mixture components. 
    The black bars are a histogram of the Hidalgo stamp data.}
  \label{fig:mu2-mode-mixtures}
\end{figure}

After computing the marginal in $\mu_2$, we tried computing the marginal in $\log_{10} \lambda_1$.  We used the natural stratification with a grid of 50 bias functions with maxima equally spaced between --1 and 3.2 constructed as
\begin{equation*}
  \psi_{i}(\theta) = \phi \left (\frac{-1 + h(i-1) - \log_{10} \lambda_1}{h} \right )
\end{equation*}
where
\begin{equation*}
  h = \frac{3.2-(-1)}{49}.
\end{equation*}
 We used the same initialization scheme as for the marginal in $\mu_2$, beginning with a single biased distribution initialized from an unbiased test trajectory.
We call this the \emph{center} sample.
The result of this calculation was the density labeled ``1D Center'' in Figure~\ref{1D_marginal_plot}.
When we tried to compute the asymptotic variance of this density estimate, we noticed very slow convergence of the sampler for some biased distributions.
To investigate, we performed another EMUS calculation using a similar initialization procedure, but starting from $\pi_1$, the biased distribution at the extreme left, covering the lowest values of $\lambda_1$.
We call this the \emph{left} sample.
The result of this second calculation was the density labeled ``1D Left'' in Figure~\ref{1D_marginal_plot}.
For both the center and left samples, the strata were equilibrated for $3000$ steps and sampled for another $200000$.
We observe that the two densities differ significantly in the region $-1 \leq \log_{10} \lambda_1 \leq 0.5$.
They should be the same up to sampling errors;
for example, we observe that different initializations have no effect on the calculation of the marginal in $\mu_2$, cf.\@ Figure~\ref{fig:mu2-marginal}.

\begin{figure}[ht]
  \subfloat[\label{1D_marginal_plot}]{\includegraphics[width=0.5\linewidth]{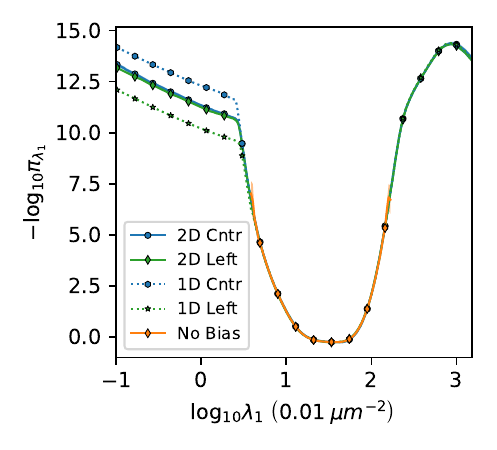}}
  \subfloat[\label{1D_marginal_err_plot}]{\includegraphics[width=0.5\linewidth]{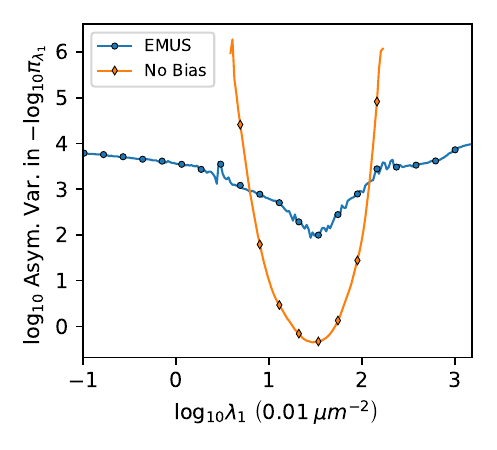}}
  \caption{Estimates of the logarithm of the marginal density in $\log_{10} \lambda_1$ and the asymptotic variances of those estimates. 
    Figure~\ref{1D_marginal_plot} displays the estimates of the marginal in $\log_{10} \lambda_1$ computed by various methods. The error bars are twice the estimated asymptotic standard deviation in each histogram bin. For the two-dimensional EMUS calculations, standard deviations were estimated using the method described in Appendix~\ref{sec:error-bars}. For both the unbiased calculation asymptotic variances were estimated using ACOR~\cite{acor}. No error bars are given for the two one-dimensional calculations, as the barrier depicted in Figure~\ref{marginal_slice} makes accurate estimation of the asymptotic variance impossible. A clear error is visible in the two one-dimensional umbrella sampling calculations, due to initialization along either side of the barrier in Figure \ref{marginal_slice}.
    Figure~\ref{1D_marginal_err_plot} displays the asymptotic variance of the marginal density in $\log_{10} \lambda_1$ for the unbiased and the two-dimensional EMUS calculations. We note that while the unbiased calculation achieves greater accuracy near the mode, the EMUS calculation achieves greater accuracy in the tails. }
\end{figure}

Figure~\ref{hist_diff} explains the problem and suggests a solution: In the region $0.2 \leq \log_{10} \lambda_1 \leq 0.7$, the center and left samples cover entirely different ranges of $\log_{10} \lambda_2$.
This suggests that the biased distributions corresponding to the range $0.2 \leq \log_{10} \lambda_1 \leq 0.7$ are multimodal, with barriers in $\lambda_2$ impeding sampling.

\begin{figure}
\begin{center}
  \includegraphics[width=0.75\linewidth]{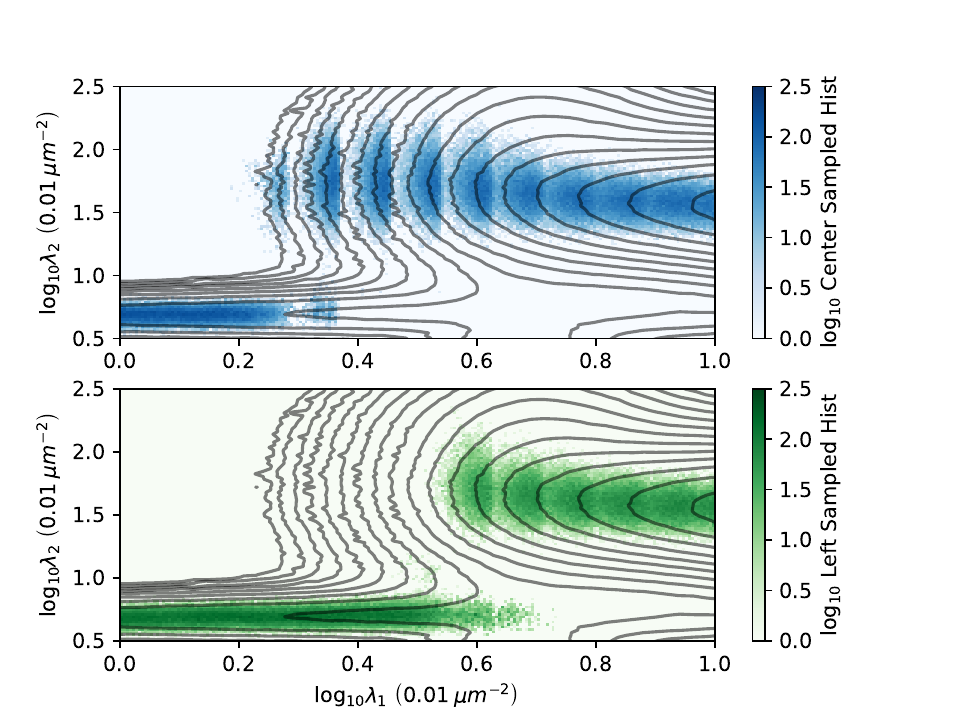}
\end{center}
\caption{
To generate Figure~\ref{hist_diff}, we binned the samples for the one-dimensional left and center EMUS calculations, and we plotted the difference in the histograms. The contour lines are contours of the log marginal density, as in Figure~\ref{2d_fe_surface}. Figure~\ref{hist_diff} shows that while the two calculations largely sample the same regions, near $\log_{10} \lambda_1=0.45$ they become trapped on opposite sides of a  barrier. This leads to poor sampling, causing a slowly decaying error in the estimates of the marginal density, cf.\@  Figure~\ref{1D_marginal_plot}.
}
\label{hist_diff}
\end{figure}

To confirm the hypothesis that barriers in $\lambda_2$ were responsible for the poor convergence observed in the center and left samples, we performed a third calculation, stratifying in both $\log_{10} \lambda_1$ and $\log_{10} \lambda_2$.
We used a $50 \times 50$  grid of bilinear bias functions, with maxima equally spaced between $-1$ and $3.2$. 
To be precise, for $i,j=1,\dots, 50$, we defined the bias functions
\begin{align*}
  \psi_{ij}(\theta) &= \phi \left (\frac{-1 + h(i-1) - \log_{10} \lambda_1}{h} \right ) \\
  &\quad \times \phi \left (\frac{-1+h(j-1)-\log_{10} \lambda_2}{h} \right ),
\end{align*}
with $h$ as before.
Let $\eta_{ij}$ denote the biased distribution corresponding to $\psi_{ij}$.
 
We performed the two-dimensional EMUS calculation twice, initializing from the center and left samples drawn from the natural stratification in $\log_{10} \lambda_1$. 
For each $i=1, \dots, L$, to sample the row $\{\eta_{ij} : j=1, \dots, 50\}$ of biased distributions, we began by initializing sampling of a single biased distribution $\eta_{ik}$ with points from the either the center or left sample of $\pi_i$.
We then sampled the other distributions $\eta_{ij}$ for $j\neq k$ in sequence, again initializing with points from samples of adjacent distributions, either $\eta_{i,j+1}$ or $\eta_{i,j-1}$ in this case. 
If no samples were found inside the support of a biased distribution, that distribution was ignored. 
For each biased distribution, sampling was burned in for 4500 steps, and samples were collected for an additional 2500 steps.  
Ultimately, 1397 of the 2500 biased distributions were sampled; the unsampled distributions correspond to the white space in Figure~\ref{2d_fe_surface}.  

\begin{figure}[ht]
  \subfloat[\label{2d_fe_surface}]{\includegraphics[width=0.5\linewidth]{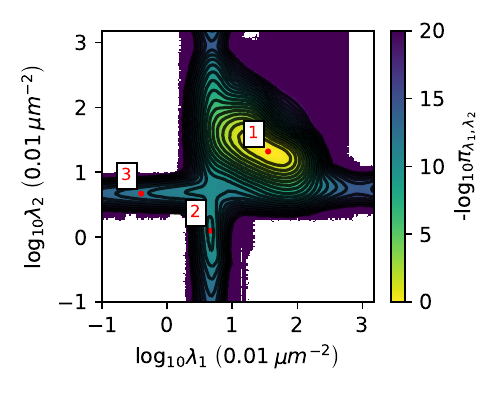}}
\subfloat[\label{unbiased}]{\includegraphics[width=0.5\linewidth]{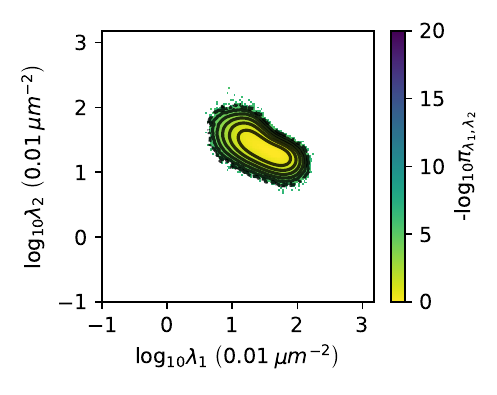}}
\caption{Logarithm of marginal density in $\log_{10} \lambda_1$ and $\log_{10} \lambda_2$ as estimated by EMUS and unbiased MCMC.
Contour lines in both figures are every unit change in the estimated $\log_{10}$ marginal density.
Figure~\ref{2d_fe_surface} is the EMUS estimate. The numbers $1$, $2$, and $3$ on this figure correspond to the mixture densities in Figure~\ref{mixtures}.  Note that at values of $\log_{10} \lambda$ near $3.0$ we begin to see the modes corresponding to singularities of the posterior. Figure~\ref{unbiased} is the marginal density estimated from a long unbiased trajectory of the ensemble sampler. Note that the entire trajectory lies in a small neighborhood of the mode labeled $1$ in Figure~\ref{2d_fe_surface}.}
\end{figure}

\begin{figure}[ht]
\begin{center}
  \includegraphics[width=0.8\linewidth]{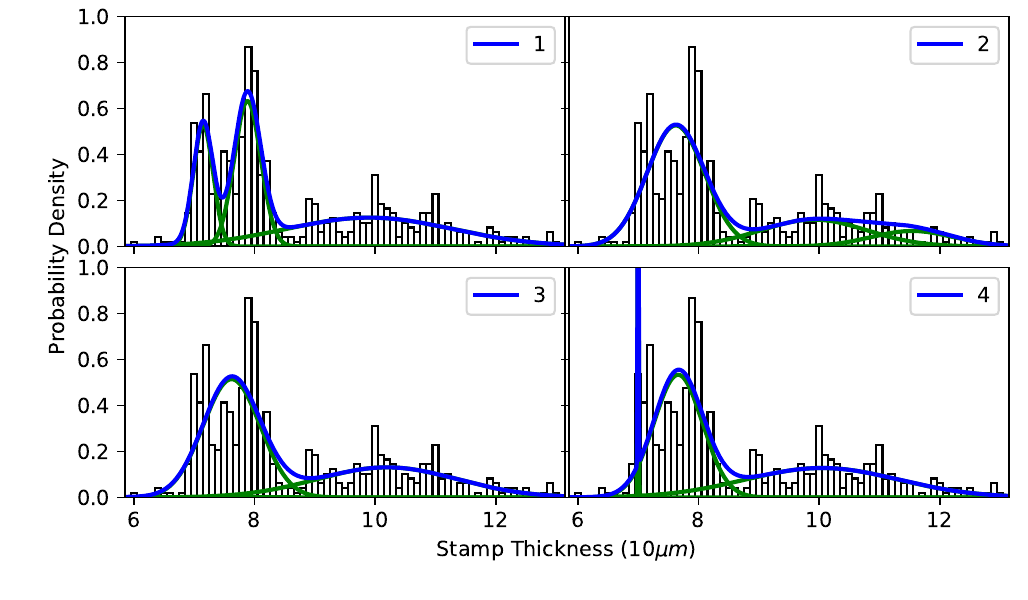}
  \caption{\label{mixtures}
  Gaussian mixtures corresponding to means of histogram bins.
  Mixtures one through three correspond to the labeled points on Figure \ref{2d_fe_surface}, mixture four corresponds to a distribution near a singularity of the posterior, with $\log_{10} \lambda_1 = 4.34$ and $\log_{10} \lambda_2 = 0.79$.
  To be precise, the blue curve in each plot is the mixture distribution corresponding to the mean of a histogram bin centered at the point labeled in Figure~\ref{2d_fe_surface}.
  The green curves are the individual mixture components. 
  The black bars are a histogram of the Hidalgo stamp data.}
\end{center}
\end{figure}

We computed the marginal in $\log_{10} \lambda_1$ and $\log_{10} \lambda_2$ using a $200 \times 200$ grid of histogram bins, covering the region ${-1 \leq \log_{10} \lambda_1 \leq 3.2}$ and ${-1 \leq \log_{10} \lambda_2 \leq 3.2}$;
this corresponds to taking ${h=(3.2-(-1))/200}$ in~\eqref{eq: defn of histogram bin}; the result from the center calculation appears in Figure~\ref{2d_fe_surface}. In Figure~\ref{mixtures}, we show the mixture distributions corresponding to the modes of the two-dimensinoal marginal in Figure~\ref{2d_fe_surface}. The two-dimensional marginals were essentially the same for the center and left initializations; see Figure~\ref{2d_fe_diff}.
We also estimated the one-dimensional marginal in $\log_{10} \lambda_1$ using the two-dimensional stratification; see the results labeled ``2D Center'' and ``2D Left'' in Figure~\ref{1D_marginal_plot}.
Finally, we estimated the relative asymptotic variance of the marginal in $\log_{10} \lambda_1$ computed by two-dimensional stratification. Again, we observe that EMUS performs much better than unbiased sampling in the tails, cf.\@ Figure~\ref{1D_marginal_err_plot}. 

\begin{figure}
\begin{center}
\includegraphics[width=0.5\linewidth]{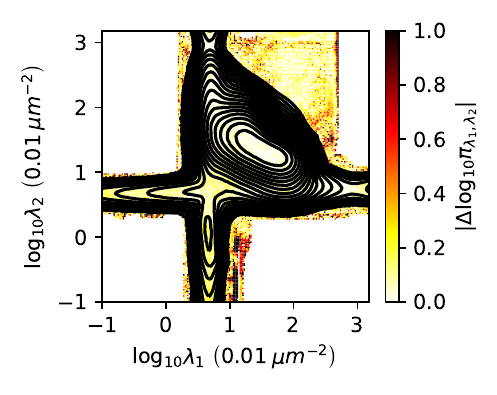}
\caption{\label{2d_fe_diff}The difference between the free energy surfaces of the two-dimensional umbrella sampling runs. The center calculation was initialized from the center one-dimensional calculation, and the left calculation from the left one-dimensional calculation. In general the difference is small, roughly a tenth of an order of magnitude in the log marginal.}
\end{center}
\end{figure}

The marginal in $\log_{10} \lambda_1$ and $\log_{10} \lambda_2$ confirms that barriers in $\lambda_2$ caused the problems observed in calculating the marginal in $\log_{10} \lambda_1$ using the natural stratification.
In fact, we see that computing the marginal in either $\lambda_1$ or $\lambda_2$ requires stratifying both variables, as stratifying only one leads to barriers that impede sampling in the other.  
In particular, there are barriers in $\lambda_2$ along the line $\log_{10} \lambda_1 = 0.45$ and a barrier in $\lambda_1$ along $\log_{10} \lambda_2 = 0.6$: In Figure~\ref{marginal_slice}, we plot an estimate of the conditional distribution of $\log_{10} \lambda_2$ with $\log_{10} \lambda_1 = 0.45$ fixed. This distribution is multimodal with a region of very low probability separating the modes, which explains the poor sampling depicted in Figure~\ref{hist_diff}.

\begin{figure}[ht]
\begin{center}
\includegraphics[width=0.5\linewidth]{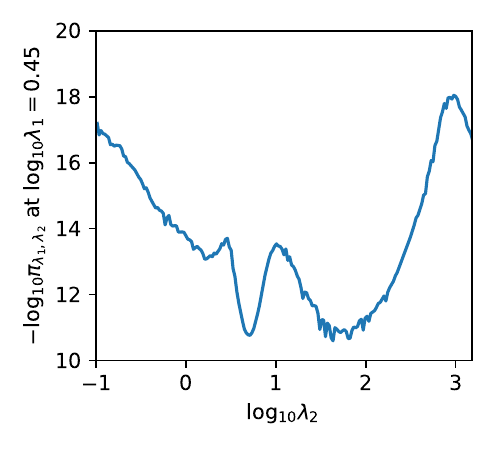}
\end{center}
\caption{
Here we give an estimate of the conditional distribution of $\log_{10} \lambda_2$ with $\log_{10} \lambda_1 = 0.45$ calculated from the two-dimensional marginal seen in Figure~\ref{2d_fe_surface}.  The conditional distribution is multimodal. The mode on the left corresponds to mixtures with the data from thicknesses of $60$ to $85$ $\mu \rm{m}$ covered by a single Gaussian similar to mode $2$ in Figure~\ref{mixtures}.
The mode on the right corresponds to mixtures with these data covered by two Gaussians similar to mode $1$ in Figure~\ref{mixtures}.}
\label{marginal_slice}
\end{figure}

To conclude, we have confirmed that EMUS can be extremely efficient for computing tails.
However, one must exercise care in the choice of strata.
The natural stratification often suffices, but in some cases, like computing the marginal in $\log_{10} \lambda_1$, the biased distributions of the natural stratification may be very difficult to sample.
We propose the use of different initializations, like the center and left samples, as a method of identifying problems related to poorly chosen strata.
Careful inspection of simulations performed with these different initializations can identify problems and suggest better strata.

\section{Conclusions}

We have analyzed the Eigenvector Method for Umbrella Sampling (EMUS), an especially simple and effective stratified MCMC method sharing many features with the popular WHAM~\cite{WHAM1992} and MBAR~\cite{shirts2008statistically} methods of computational chemistry. 
We have demonstrated the advantages of EMUS for sampling from multimodal distributions and computing tail probabilities, and we have explained how to identify and resolve the problems which may occur if the method is implemented poorly. 
We have also given a tutorial intended to explain how to diagnose and correct problems related to poorly chosen strata.

Our purpose was to explain the benefits of stratified MCMC analytically, with the ultimate goal of introducing stratified MCMC to a diverse audience of statisticians, engineers, and scientists. Since stratified MCMC had previously been applied only to a particular class of statistical mechanics calculations without any general justification, we began by developing a general theory. We hope that our theory will serve as the basis for further developments. For example, it may now be possible to undertake a comparison of EMUS and other so-called reaction coordinate methods such as Wang--Landau sampling~\cite{wang2001efficient} or Metadynamics~\cite{laio2002escaping}. Despite some similarities with EMUS, these methods work by a substantially different mechanism and understanding the relative advantages of the two approaches is non-trivial.
We also note that there is potential to apply stratification to problems that lie outside the scope of the present work. For example, we present a stratification method capable of computing dynamical quantities such as mean first passage times~\cite{dinner_trajectory_2018}.

\appendix

\section{Derivation of~\eqref{eq: decomposition of average 2} and~\eqref{eq: eigenvector problem}}
\label{apx: derivation of emus}

To see that~\eqref{eq: decomposition of average 2} holds, observe that
\begin{align*}
  \sum_{i=1}^L z_i \pi_i[g^\ast]/u_i &= \sum_{i=1}^L \frac{1}{\sum_{k=1}^L \pi[\psi_k]} \pi \left [ \frac{g \psi_i/u_i}{ \sum_{k=1}^L \psi_k/u_k} \right ] \\
                                     &= \frac{\pi[g]}{\sum_{k=1}^L \pi[\psi_k]}.
\end{align*}
Therefore, we have
\begin{equation*}
  \frac{\sum_{i=1}^L z_i \pi_i[g^\ast]/u_i}{\sum_{i=1}^L z_i \pi_i[\1^\ast]/u_i}= \frac{\pi[g]/ \left(\sum_{k=1}^L \pi[\psi_k] \right)}{\pi[\1]/ \left(\sum_{k=1}^L \pi[\psi_k] \right)} = \pi[g].
\end{equation*}

To prove~\eqref{eq: eigenvector problem}, observe that
\begin{align*}
  \sum_{i=1}^L w_i F_{ij} &= \sum_{i=1}^L \frac{z_i}{u_i} \frac{\pi_i[\psi_j^\ast]}{u_j} \\
                          &= \sum_{i=1}^L \frac{1}{u_j}\pi \left [ \frac{\psi_j (\psi_i/u_i)}{\sum_{k=1}^L \psi_k/u_k} \right ] \\
                          &=\frac{\pi[\psi_j]}{u_j} \\
  &= w_j.
\end{align*}

\section{Proof of Lemma~\ref{lem: irreducibility condition}}
\label{apx: irreducibility condition}
\begin{proof}
We prove only the second statement; proof of the first is similar.
By definition, a non-negative matrix $M \in \Real^{L \times L}$ is irreducible if and only if for every subset $A \subset \{1,2,\dots,L\}$ of the indices, there exist indices $i \in A$ and $j \notin A$ so that $M_{ji} > 0$. 
Now assume that for every $A \subset \{1,2,\dots, L\}$, there exist $j \notin A$ and $t \geq 0$ so that 
\begin{equation*}
 X^j_t \in \cup_{k \in A} \{ x: \psi_k(x) >0 \}.
\end{equation*}
Then for some $i \in A$, $\psi_i(X^j_t) > 0$, so $\bar F_{ji} >0$, hence $\bar F$ is irreducible.
\end{proof}

\section{Sparse Grid of Strata}
\label{apx: sparse grid}    
One may define a uniform grid of strata so that~\eqref{eq: key factor in section on relative variances} increases only as $d^2$ with dimension, not exponentially: For any $\iidx \in \Z^d$, let $V_\iidx' := h \iidx + h \left [-\frac12,\frac12 \right ]^d$. For $\iidx \neq \jidx$, define
  \begin{align*}
    W_{\iidx \jidx} := \Big \{ x \in V'_\jidx : \min_{y \in V'_\iidx} \lVert x -y\rVert &\leq \min_{y \in V'_\kidx} \lVert x -y\rVert \text{ for any } \kidx \in \Z^d\setminus \{\jidx\} \  \Big \}
  \end{align*}
  to be the $d$-dimensional pyramid consisting of all points in $V'_\jidx$ closer to $V'_\iidx$ than to any other cube $V'_\kidx$.
   Now let $e_n$ denote the $n$'th standard basis vector in $\Real^d$, and define
    \begin{equation*}
      V_\iidx := \cup_{n=1}^d (W_{\iidx, \iidx+e_n} \cup W_{\iidx, \iidx-e_n}) \cup V_\iidx'
    \end{equation*}
    to be the cube $V_\iidx'$ enlarged by all the neighboring pyramids $W_{\iidx\jidx}$.
    The strata $V_\iidx$ are convex, and the corresponding bias functions $\psi_\iidx  = \frac12 \1_{V_\iidx}$ are a partition of unity.
    Each stratum $V_{\iidx}$ intersects only the $2d$ neighboring strata $V_{\iidx \pm e_n}$ for $n=1,\dots, d$.
    Moreover, each intersection between neighboring strata $V_\iidx$ and $V_\jidx$ consists of the pair of pyramids $W_{\iidx\jidx}$ and $W_{\jidx\iidx}$, and it has volume $1/d$.
    Therefore, by Lemma~\ref{lem: size of entries in F}, for this choice of bias functions, the nonzero entries of $F$ decrease as $1/d$.
    It follows that~\eqref{eq: key factor in section on relative variances} increases as $d^2$.

\section{Proofs of Theorem~\ref{thm: CLT for EMUS} and Theorem~\ref{thm: upper bound on asymptotic variance}}
\label{apx: proof of clt}
Our proof of Theorem~\ref{thm: CLT for EMUS} (the CLT for EMUS) is based on the delta method. To apply the delta method, we require the following result ensuring the differentiability of $w(G)$:
\begin{lemma}\label{lem: differentiability of z}
The function $w(G)$ admits an extension $\tilde{w}: \Real^{L \times L} \rightarrow \Real^L$ which is differentiable on the set of irreducible stochastic matrices.
\end{lemma}

\begin{proof}
By~\cite[Lemma~3.1]{ThVKWe:Perturbation}, $w(G)$ admits a continuously differentiable extension to an open set $U \subset \Real^{L \times L}$.
We further extend the domain of $w(G)$ to $\Real^{L \times L}$ by arbitrarily defining $w(G) = 0$ whenever $G \in \Real^{L \times L} \setminus U$.
\end{proof}

The extension in Lemma~\ref{lem: differentiability of z} resolves two technicalities:
First, the set of stochastic matrices is not a vector space
but a compact, convex subset of $\Real^{L \times L}$ with empty interior.
Therefore, the derivative of $w$ is undefined.
Second, $\bar F$ may be reducible for some values of $N$
and some realizations of the processes sampling the biased distributions. In that case,
the invariant distribution of $\bar F$ is not unique, so $w(\bar F)$ is undefined.
Throughout the remainder of this work, $w(G)$ will denote the extension guaranteed by the lemma.

We now prove the CLT for EMUS.

\begin{proof}[Proof of Theorem~\ref{thm: CLT for EMUS}]
The proof is based on the delta method~\cite[Proposition~6.2]{BrennerBilodeau:TheoryMultivariateStats} and a formula for $w'(\bar F)$ given in~\cite{GolMey:ComputingInvDist}.

By Lemma~\ref{lem: differentiability of z}, $w(\bar F)$ is differentiable at $F$, so the function 
\begin{equation*}
 B \left (\bar F, \{ \bar g_i^\ast \}_{i=1}^L, \{\bar \1_i^\ast\}_{i=1}^L \right ) 
 := \pi_\emus[g] = \frac{\sum_{i=1}^L w_i(\bar F) \bar g^\ast_i}{\sum_{i=1}^L w_i(\bar F) \bar \1^\ast_i}
\end{equation*}
is differentiable at 
$\left (F, \{\pi_i[g^\ast]\}_{i=1}^L, \{\pi_i[\1^\ast]\}_{i=1}^L\right )$.
Let $\partial_i B \in \Real^{L+2}$ be the derivative of $B$ with respect to those quantities computed from $X^i_t$: That is, 
\begin{equation}
 \partial_i B := \left ( 
 \frac{\partial B}{\partial \bar F_{i:}},
 \frac{\partial B}{\partial \bar G_{i:}}
 \right ) \in \Real^{L+2},
 \label{eq: definition of derivatives of B}
\end{equation}
where $\frac{\partial B}{\partial \bar F_{i:}} \in \Real^L$ denotes the partial derivative of $B$ with respect to the $i$'th row of $\bar F$ and 
\begin{equation*}
 \frac{\partial B}{\partial \bar G_{i:}}  = \left ( \frac{\partial B}{\partial \bar g^\ast_i},\frac{\partial B}{\partial \bar \1^\ast_i}  \right ) \in \Real^2.
\end{equation*}
To simplify notation, we will assume throughout the remainder of this argument that all derivatives are evaluated at 
$\left (F, \{\pi_i[g^\ast]\}_{i=1}^L, \{\pi_i[\1^\ast]\}_{i=1}^L  \right )$.
In formulas involving matrix multiplication, we will treat $\partial_i B$, $\frac{\partial B}{\partial \bar F_{i:}}$, and $\frac{\partial B}{\partial \bar G_{i:}}$ as row vectors.

Since we assume that the processes $X^i_t$ sampling the different measures $\pi_i$ are independent, \cite[Chapter~1, Theorem~2.8]{Billingsley:ConvergenceProbabilityMeasures} implies that

\begin{align}\label{eq: central limit theorem for F and f}
  &\sqrt{M} \Big ( \left ( \bar{F}_{1:}, \bar{g}_1, \bar \1^\ast_1, \dots,
    \bar{F}_{L:}, \bar g^\ast_L, \bar \1^\ast_L \right ) - \left ( F_{1:}, \pi_1[g^\ast], \pi_1[\1^\ast],
    \dots,F_{L:}, \pi_L[g^\ast], \pi_L[\1^\ast] \right ) \Big ) \convdist N (0,\Sigma),
\end{align}
where $\Sigma$ is the covariance matrix of the product of the distributions $N \left (0, \kappa_i^{-1} \Sigma_i \right )$.
(That is, ${\Sigma \in \Real^{L(L+2) \times L(L+2)}}$ is the block diagonal matrix with the matrices $\kappa_i^{-1} \Sigma_i$ along the diagonal.)
Therefore, by the delta method,
\begin{equation*}
 \sqrt{M} ( \pi_\emus [g] - \pi[g] ) \convdist \mathrm{N}(0, \sigma^2),
\end{equation*} 
where 
\begin{align}
 \sigma^2 &= (\partial_1 B, \dots, \partial_L B) \Sigma (\partial_1 B, \dots, \partial_L B)^\t \nonumber \\
 &= \sum_{i=1}^L\kappa_i^{-1} \partial_i B \Sigma_i \partial_i B^\t.
 \label{eq: CLT proof, first formula for asymp var}
\end{align}

Now we observe that for any column vector $v \in \Real^L$ having mean zero, 
\begin{equation*}
\left . \frac{d}{d \eps} w_k(F + \eps e_i v^\t)  \right \rvert_{\eps=0} = \frac{\partial w_k}{\partial \bar F_{i:}} v = z_i v^\t (I-F)^\# e_k ,
\end{equation*}
by~\cite[Theorem~3.1]{GolMey:ComputingInvDist}. 
(In the formula above, $e_i \in \Real^L$ denotes the $i$'th standard basis vector.)
Therefore, we have
\begin{align}
 \frac{\partial B}{\partial \bar F_{i:}} v
 &= 
 \frac{\sum_{k=1}^L 
 \frac{\partial w_k}{\partial \bar F_{i:}}v \pi_k[g^\ast]}{\sum_{k=1}^L z_k \pi_k[\1^\ast]} 
 -
 \frac{\sum_{k=1}^L 
 \frac{\partial w_k}{\partial \bar F_{i:}}v \pi_k[\1^\ast]}
 {\sum_{i=1}^L z_k \pi_k[\1^\ast]} 
 \frac{\sum_{i=1}^L z_k \pi_k[g^\ast]}{\sum_{k=1}^L z_k \pi_k[\1^\ast]} \nonumber \\
 &=
 \sum_{k=1}^L 
 \frac{\partial w_k}{\partial \bar F_{i:}}v \Psi (\pi_k[ g^\ast] - \pi[g] \pi_k[\1^\ast]) 
 \label{eq: delta method proof, formula for derivative wrt F} \\
 &=
 z_i v^\t (I-F)^\# \mathfrak{g},
 \nonumber
\end{align}
where 
\begin{equation*}
\mathfrak{g}_k =\Psi \pi_k \left [g^\ast - \pi[g] \1^\ast \right] = \ell \cdot (\pi_k[g^\ast], \pi_k[\1^\ast]).
\end{equation*}
(Equality~\eqref{eq: delta method proof, formula for derivative wrt F} above follows from~\eqref{eq: decomposition of average 2} and the definition~\eqref{eq: definition of Psi} of $\Psi$.)
Also, 
\begin{equation}\label{eq: delta method proof, formula for derivative wrt g}
\frac{\partial B}{\partial \bar G_i}= z_i \Psi (1, -\pi[g]) = z_i \ell.
\end{equation}
Thus, 
\begin{align*}
  \partial_i B \Sigma_i \partial_i B^\t &= \frac{\partial B}{\partial \bar F_{i:}} \sigma^i \frac{\partial B}{\partial \bar F_{i:}}^\t + 2 \frac{\partial B}{\partial \bar F_{i:}} \rho_i \frac{\partial B}{\partial \bar G_{i:}}^\t + \frac{\partial B}{\partial \bar G_{i:}} \tau_i \frac{\partial B}{\partial \bar G_{i:}} \\
                                        &= z_i^2 \Big \{ (I-F)^\# \mathfrak{g}  \cdot \sigma^i (I-F)^\# \mathfrak{g} + 2 (I-F)^\# \mathfrak{g}  \cdot \rho_i \ell + \ell^\t \tau_i \ell \Big \},
\end{align*}
and the result follows by~\eqref{eq: CLT proof, first formula for asymp var}.
\end{proof}

We now prove of Theorem~\ref{thm: upper bound on asymptotic variance}. To begin, we give some upper bounds on the partial derivatives of the weight vector $w(F)$ with respect to the entries of the overlap matrix $F$.

\begin{definition}\label{def: log partial der}
Let $e_i \in \Real^L$ denote the $i$'th standard basis vector. For $i,j \in \{1,2, \dots, L\}$ with $i \neq j$, define the logarithmic partial derivatives
\begin{align}\label{eq: defn of log partial deriv}
  \frac{\partial \log w_k}{\partial {F}_{ij}}(F) :=&
   \frac{\partial}{\partial F_{ij}}
  \log w_k \left (\sum_{i \neq j} I + F_{ij} \left (e_i e_j^\t - e_i e_i^\t\right ) \right )
  \nonumber \\
  =&
  \left. \frac{d}{d\eps} \right \rvert_{\eps=0}
  \log w_k(F + \eps (e_ie_j^\t - e_i e_i^\t)).
\end{align}
(These partial derivatives must be understood as derivatives of the extension
guaranteed by Lemma~\ref{lem: differentiability of z}; otherwise, they are defined only when
$F_{ij} > 0$ and $F_{ii} >0$.)
\end{definition}

Our definition of logarithmic partial derivatives in~\eqref{eq: defn of log partial deriv} is not standard.
However, we observe that a version of the standard formula relating the total and partial derivatives of $\log w$ holds: For all matrices $H$ whose rows sum to zero,
\begin{equation}\label{eq: total derivative in terms of partial derivatives}
 \left . \frac{d}{d\eps} \right \rvert_{\eps =0} \log w_k (F + \eps H) =
 \sum_{i \neq j} \frac{\partial \log w_k}{\partial F_{ij}}(F) H_{ij}.
\end{equation}
We need only consider matrices whose rows sum to zero, since these are the only perturbations for which $F + \eps H$ can be stochastic.

The following result appears in~\cite[Theorem~3.6]{ThVKWe:Perturbation}. 
It is crucial in our proof of Theorem~\ref{thm: upper bound on asymptotic variance}.

\begin{lemma}\label{lem: bound on logarithmic derivatives}
Recall $\P_i[t_j < t_i]$ and $\frac{\partial \log w_k}{\partial {F}_{ij}}$
from Definitions~\ref{def: definitions of pij} and~\ref{def: log partial der}.
For all stochastic and irreducible matrices $F$, we have
\begin{equation*}
\frac{1}{2} \frac{1}{\P_i [t_j < t_i]} \leq
 \max_k \left \lvert \frac{\partial \log w_k}{\partial {F}_{ij}}(F) \right \rvert
\leq \frac{1}{\P_i [t_j < t_i]}.
\end{equation*}
\end{lemma}

We also require the following lemma in the proof of Theorem~\ref{thm: upper bound on asymptotic variance}.

\begin{lemma}\label{lem: properties of covariance matrix}
The asymptotic covariance matrix $\sigma^i$ has the properties:
\begin{enumerate}
\item The rows and colums of $\sigma^i$ sum to zero. That is, for $e \in \Real^L$ the vector of all ones, 
\begin{equation*}
 \sigma^i e= 0 \text{ and } e^\t \sigma^i =0.
\end{equation*}
\item 
 For all $j = 1, \dots, L$, 
 \begin{equation*}\label{eq: sparsity of asymptotic cov}
  \sigma^i_{jk}=\sigma^i_{kj} = 0 \text{ whenever } F_{ik} = 0.
 \end{equation*}
\end{enumerate}
\end{lemma}

\begin{proof}
Since the rows of $\bar F$ sum to one with probability one, we have 
\begin{equation*}
 \var (\bar F_{i:} e) = 0
\end{equation*}
for any fixed number of samples $N_i$. Therefore, the asymptotic variance $\sigma^i$ has $e^\t \sigma^i e = 0$, and it follows that $e^\t \sigma^i = \sigma^i e = 0$ since $\sigma^i$ is symmetric and positive semidefinite.

Let $k$ be such that $F_{ik} = 0$. Since $\bar F_{ik}=0$ with probability one, we have 
\begin{equation*}
 \cov (\bar F_{ik}, \bar F_{ij}) = 0
\end{equation*}
for any $j = 1, \dots, L$, and therefore $\sigma^i_{jk} = 0$. 
\end{proof}

We now prove Theorem~\ref{thm: upper bound on asymptotic variance}.
\begin{proof}[Proof of Theorem~\ref{thm: upper bound on asymptotic variance}]
  We begin with formula~\eqref{eq: CLT proof, first formula for asymp var}:
  \begin{equation}\label{eq: restatement of basic formula for asymptotic variance of emus}
    \sigma^2= \sum_{i=1}^L\kappa_i^{-1} \partial_i B^\t \Sigma_i \partial_i B.
  \end{equation}
  Since the asymptotic covariance matrix $\Sigma_i$ is symmetric and positive semidefinite, the Cauchy inequality holds:
  \begin{equation*}
    a^\t \Sigma_i b \leq \frac12 a^\t \Sigma_i a + \frac12 b^\t \Sigma_i b,  
  \end{equation*}
  for all $a,b \in \Real^{L+1}$. 
  Therefore,
  \begin{align}
    \partial_i B^\t \Sigma_i \partial_i B 
    &=
      \left (\frac{\partial B}{\partial \bar F_{i:}}, \frac{\partial B}{\partial\bar G_{i:}} \right ) 
      \Sigma_i 
      \left (\frac{\partial B}{\partial \bar F_{i:}}, \frac{\partial B}{\partial\bar G_{i:}} \right )^\t
      \nonumber \\ 
    &\leq 2 \left (\frac{\partial B}{\partial \bar F_{i:}}, 0 \right ) 
      \Sigma_i 
      \left (\frac{\partial B}{\partial \bar F_{i:}}, 0 \right )^\t
      + 
      2 \left (\mathbf{0}^\t, \frac{\partial B}{\partial\bar G_{i:}} \right )
      \Sigma_i
      \left (\mathbf{0}^\t, \frac{\partial B}{\partial\bar G_{i:}} \right )^\t 
      \nonumber \\ 
    &= 2 \frac{\partial B}{\partial \bar F_{i:}} \sigma^i \frac{\partial B}{\partial \bar F_{i:}}^\t + 2 \frac{\partial B}{\partial \bar G_{i:}} \tau_i \frac{\partial B}{\partial \bar G_{i:}}^\t. \nonumber \\
    &=: 2A_0 + 2A_1 . \label{eq: cauchy inequality in proof of upper bound}
  \end{align}
  (Here, $\mathbf{0}$ denotes the zero vector in $\Real^L$, interpreted as a column vector.)

  We now estimate the term $A_0$ defined above.
  By~\eqref{eq: delta method proof, formula for derivative wrt F}, we have 
  \begin{align}
    A_0 &= \frac{\partial B}{\partial \bar F_{i:}}^\t \sigma^i \frac{\partial B}{\partial \bar F_{i:}} \nonumber \\
        &= \sum_{j,k,\ell,m=1}^L \mathfrak{g}_\ell \frac{\partial w_\ell}{\partial\bar F_{ij}} \sigma^i_{jk}  \mathfrak{g}_m \frac{\partial w_m}{\partial\bar F_{ij}} \nonumber \\
        &= \sum_{\ell, m=1}^L z_\ell \mathfrak{g}_\ell z_m \mathfrak{g}_m \sum_{\substack{j \neq i \\ F_{ij} >0}}
    \sum_{\substack{k \neq i \\ F_{ik} >0}} 
    \frac{\partial \log w_\ell}{\partial \bar F_{ij}} \sigma^i_{jk} \frac{\partial \log w_m}{\partial \bar F_{ik}}.
    \nonumber \\
        &= \sum_{\ell, m=1}^L z_\ell \mathfrak{g}_\ell z_m \mathfrak{g}_m
          \sum_{\substack{j \neq i \\ F_{ij} >0}}
    \sum_{\substack{k \neq i \\ F_{ik} >0}} 
    \sqrt{\var_{\pi_i}(\psi_j^\ast)}  
    \frac{\partial \log w_\ell}{\partial \bar F_{ij}} 
    R^i_{jk}
    \sqrt{\var_{\pi_i}(\psi_k^\ast)} 
    \frac{\partial \log w_m}{\partial \bar F_{ik}},
    \label{eq: a0 with autocorrelation}
  \end{align}
  where 
  \begin{equation*}
    R^i_{jk}:=
    \frac{\sigma^i_{jk}}{\sqrt{\var_{\pi_i}(\psi_j^\ast)}\sqrt{\var_{\pi_i}(\psi_k^\ast)}}.
  \end{equation*}
  (The third equality above follows from formula \eqref{eq: total derivative in terms of partial derivatives} relating the total and partial derivatives of $\log w$, since the rows and colums of $\sigma^i$ sum to zero by Lemma~\ref{lem: properties of covariance matrix}.)

  We claim that 
  \begin{equation}
    \begin{split}  
    &\sum_{\substack{j \neq i \\ F_{ij} >0}}
    \sum_{\substack{k \neq i \\ F_{ik} >0}} 
    \sqrt{\var_{\pi_i}(\psi_j^\ast)}  
    \frac{\partial \log w_\ell}{\partial \bar F_{ij}} 
    R^i_{jk}
    \sqrt{\var_{\pi_i}(\psi_k^\ast)} 
    \frac{\partial \log w_m}{\partial \bar F_{ik}}
    \leq \tr(R^i) \sum_{\substack{j \neq i \\ F_{ij} >0}}
    \var_{\pi_i}(\psi_j^\ast) 
    \frac{\partial \log w_\ell}{\partial \bar F_{ij}}^2.
    \label{eq: claimed trace estimate}
  \end{split}
\end{equation}
  To prove this, we observe that $R^i$ is symmetric and positive semidefinite since $\sigma^i$ is symmetric and positive semidefinite. Therefore, $R^i$ has the spectral decomposition 
  \begin{equation*}
    R_i = \sum_{j=1}^L \lambda_{i,j} v^{i,j} (v^{i,j})^\t
  \end{equation*}
  with eigenvalues $\lambda_{i,j}>0$ and corresponding eigenvectors $v^{i,j}$ such that $\lVert v^{i,j} \rVert=1$.
  Thus, for any  $a \in \Real^L$, 
  \begin{align}
    a^\t R^i a &= \sum_{j=1}^L \lambda_{i,j} \lvert v^{i,j} \cdot a \rvert^2 
                 \leq  \left (\sum_{j=1}^L \lambda_{i,j} \right ) \lVert a \rVert^2 = \tr(R^i) \lVert a \rVert^2.
                 \label{eq: general trace estimate}
  \end{align}
  Inequality~\eqref{eq: claimed trace estimate} follows from~\eqref{eq: general trace estimate} by setting
  \begin{equation*}
    a_j =
    \begin{cases}
      \sqrt{\var_{\pi_i}(\psi_j^\ast)}  
      \frac{\partial \log w_\ell}{\partial \bar F_{ij}} 
      &\text{ if } 
      j \neq i \text{ and } F_{ij} >0, \text{ and } \\
      0 &\text{ otherwise.}
    \end{cases}
  \end{equation*}

  Finally, combining~\eqref{eq: a0 with autocorrelation},~\eqref{eq: claimed trace estimate}, and Lemma~\ref{lem: bound on logarithmic derivatives} yields
  \begin{equation*}
    A_0 \leq 
    \tr(R^i) \left ( \sum_{\ell=1}^L z_\ell \lvert \mathfrak{g}_\ell \rvert \right )^2  
    \sum_{\substack{j \neq i \\ F_{ij} >0}}
    \frac{\var_{\pi_i}(\psi_j^\ast)}{\P_{i}[t_j < t_i]^2}.
  \end{equation*}
  Moreover, we have 
  \begin{align*}
    \sum_{\ell=1}^L z_\ell \lvert \mathfrak{g}_\ell \rvert 
    &=
      \Psi \sum_{\ell=1}^L z_\ell \lvert \pi_\ell [ g^\ast - \pi[g] \1^\ast ] \rvert 
      \leq \Psi \sum_{\ell=1}^L z_\ell \pi_\ell [ \lvert h \rvert ] 
      = \pi[\lvert h \rvert],
  \end{align*}
  by~\eqref{eq: decomposition of average 2}, and therefore
  \begin{equation}
    \label{eq: final estimate of a0}
    A_0  
    \leq 
    \tr(R^i)\pi[\lvert h \rvert ]^2  
    \sum_{\substack{j \neq i \\ F_{ij} >0}}
    \frac{\var_{\pi_i}(\psi_j^\ast)}{\P_{i}[t_j < t_i]^2}.
  \end{equation}

  We now observe that by~\eqref{eq: delta method proof, formula for derivative wrt g} 
  \begin{equation*} 
    A_1 = z_i^2 \ell^\t \tau_i \ell = \Psi^2 z_i^2 \acov (\bar h_i),
  \end{equation*}
  $\acov ( \bar h_i)$ denotes the asymptotic covariance of the trajectory average $\bar h_i$ of $h$ over the biased process $X^i_t$.
  Therefore, combining~\eqref{eq: restatement of basic formula for asymptotic variance of emus} and~\eqref{eq: final estimate of a0}, we find 
  \begin{equation*}
    \sigma^2 \leq 2\sum_{i=1}^L \frac{1}{\kappa_i}
    \left \{
      z_i^2 \Psi^2 \acov(\bar h_i) 
      + \tr(R^i)\pi[\lvert h \rvert ]^2  
      \sum_{\substack{j \neq i \\ F_{ij} >0}}
      \frac{\var_{\pi_i}(\psi_j^\ast)}{\P_{i}[t_j < t_i]^2}
    \right \},
  \end{equation*}
  as desired.
\end{proof}

When the bias functions are a partition of unity, both the EMUS method and the statements of Theorems~\ref{thm: CLT for EMUS} and~\ref{thm: upper bound on asymptotic variance} simplify considerably. (The bias functions are a partition of unity if and only if $\sum_{i=1}^L \psi_i(x) =1$ for all $x$.) In this case, $f^\ast =f$ for all functions $f$, and the EMUS method reduces to 
\begin{align*}
  \pi_\emus [g]&=\sum_{i=1}^L w_i(\bar F) \bar g_i, 
\end{align*}
where
\begin{align*}
\bar F_{ij} &= N_i^{-1} \sum_{t=1}^{N_i} \psi_j(X^i_t) \text{ and } 
  \bar g_i = N_i^{-1} \sum_{t=1}^{N_i} g(X^i_t).
\end{align*}

\begin{corollary}\label{cor: upper bound, partition of unity}
Suppose that the bias functions are a partition of unity. In that case, Theorem~\ref{thm: upper bound on asymptotic variance} holds with either $\var_\pi(g)$ or $\pi[\lvert g \rvert ]^2$ in place of $\pi[\lvert h \rvert]^2$. In addition, $\Psi=1$, and one can replace $\acov(\bar h_i)$ with the asymptotic variance $\acov (\bar g_i)$ of $\bar g_i$. 
\end{corollary}

\begin{proof}
  When the bias functions are a partition of unity,
  \begin{equation*}
    \pi[\lvert h \rvert]^2 = \pi[\lvert g - \pi[g] \rvert ]^2 \leq  \pi[\lvert g - \pi[g] \rvert^2 ] = \var_\pi(g),
  \end{equation*}
  and so we may replace $\pi[\lvert h \rvert]^2$ with $\var_\pi(g)$. 
  In addition, equation~\eqref{eq: delta method proof, formula for derivative wrt F} holds with $\mathfrak{g}_k = \pi_k[g]$. 
  Thus, following the argument above, one may verify that the result also holds with $\pi[\lvert g \rvert ]^2$ in place of $\pi[\lvert h \rvert]^2$.
\end{proof}


\section{Proof of Theorem~\ref{thm: ergodicity and estimates of iat}}
\label{apx: estimate of integrated autocovariance}

In the arguments below, for any probability measure $\nu$ on a set $\Omega$, we let 
\begin{equation*}
 L^2(\nu) := \{ u: \Omega \rightarrow \Real : \nu[u^2] < \infty\},
\end{equation*}
and we define the $L^2(\nu)$ inner product 
\begin{equation*}
 \langle f,g \rangle_\nu = \nu[fg]
\end{equation*}
with the corresponding norm
\begin{equation*}
 \lVert f \rVert_{L^2(\nu)} := \sqrt{\langle f,f\rangle_\nu}.
\end{equation*}
Given a set $U \subset \Real^d$, we define $L^2(U)$, $\lVert \cdot \rVert_{L^2(U)}$, $\langle , \rangle_U$ to be the analogous function space, norm, and inner product for Lebesgue measure on $U$.

Our proof of Theorem~\ref{thm: ergodicity and estimates of iat} requires a Poincar\'{e} inequality, Lemma~\ref{lem: poincare inequality}. 
We refer to~\cite[Section~3]{LelievreStoltz:PDEinMD} for an introduction to Poincar\'e inequalities and their role in the theory of diffusion processes.

\begin{lemma}\label{lem: poincare inequality}
  Assume that the Poincar\'e inequality holds for $U$ with constant $\Lambda$; that is, assume that for all weakly differentiable $f: U \rightarrow \Real$ so that $\nabla f \in L^2(U)$, 
  \begin{equation*}
  \left  \lVert f - \int_U f \, dx \right \rVert_{L^2(U)} \leq 
    \Lambda(U) \lVert \nabla f\rVert_{L^2(U)}
  \end{equation*}
 We have a similar Poincar\'e inequality for $\pi_h$:
 \begin{equation*}
   \begin{split}
     &\lVert f - \pi_h(f) \rVert_{L^2(\pi_h)}  \leq 
    h \Lambda(U)
    \exp\left (\frac{\beta}{2} \left (  \sup_{U_h} V - \inf_{U_h} V \right ) \right ) \lVert \nabla f\rVert_{L^2(\pi_h)}.
  \end{split}
\end{equation*}
\end{lemma}
\begin{proof}
By a standard scaling argument, the Poincar\'e inequality holds for $U_h$ with constant $h \Lambda$. 
To see this, let $A_h : U \rightarrow U_h$ be the affine transformation
\begin{equation*}
  A_h x = x_0 + h (x-x_0).
\end{equation*}
For any $f: U_h \rightarrow \Real$ with $\nabla f \in L^2(U_h)$, using the change of variable formula and the chain rule, we have 
\begin{align*}
 \left \lVert f - \int_{U_h} f \right \rVert^2_{L^2(U_h)} 
 &= 
 h^d \left \lVert f \circ A_h - \int_U f \circ A_h  \right \rVert^2_{L^2(U)} \\
 &\leq 
 h^d \Lambda^2 
 \left \lVert \nabla (f \circ A_h)  \right \rVert^2_{L^2(U)} \\
 &= h^d h^2 \Lambda^2
 \left \lVert (\nabla f )\circ A_h  \right \rVert^2_{L^2(U)} \\
 &= h^2 \Lambda^2 
 \left \lVert \nabla f \right \rVert^2_{L^2(U_h)}. 
\end{align*}

Now observe that for any $f \in L^2(\pi_h)$,
\begin{equation*}
\lVert f - \pi_h[f] \rVert_{L^2(\pi_h)} = \min_{c \in \Real} \lVert f - c\rVert_{L^2(\pi_h)},
\end{equation*}
since $\pi_h [f]$ is the $L^2(\pi_h)$ orthogonal projection of $f$ onto the space of constant functions.
Therefore, we have
\begin{align*}
\lVert f - \pi_h[f] \rVert_{L^2(\pi_h)}^2
&\leq \left \lVert f - \int_{U_h} f \right \rVert_{L^2(\pi_h)}^2 \\
&\leq \left ( \sup_{x \in U_h} \pi_h(x) \right )
\left \lVert f - \int_{U_h} f \right \rVert_{L^2(U_h)}^2 \\
&\leq h^2 \Lambda^2 \left ( \sup_{U_h} \pi(x)  \right )
\lVert \nabla f \rVert_{L^2(U_h)}^2 \\
&\leq h^2 \Lambda^2 
\frac{\sup_{x \in U_h} \pi_h (x)}{\inf_{x \in U_h} \pi_h(x)}
\lVert \nabla f \rVert_{L^2(\pi_h)}^2,
\end{align*}
and the result follows.
\end{proof}

\begin{remark}
 The Poincar\'e inequality for the Lebesgue measure on a set $U$ holds under very weak conditions on $U$. For example, when $U$ is convex, the Poincar\'e inequality holds with constant $\Lambda (U) = D/\pi$, where $D$ is the diameter of the domain~\cite{payne1960poincare}.
\end{remark}

We now prove Theorem~\ref{thm: ergodicity and estimates of iat}:

\begin{proof}[Proof of Theorem~\ref{thm: ergodicity and estimates of iat}]
We begin by stating a simple consequence of the functional central limit theorem for reversible, continuous time Markov processses:
Let $Y_t$ be a reversible, stationary Markov process with ergodic distribution $\pi$ and generator $L$.
Let $g \in L^2(\pi)$, and define
\begin{equation*}
 \bar g := T^{-1} \int_{s=0}^T g(Y_s)  \, ds.
\end{equation*}
By~\cite[Corollary~1.9]{kipnis1986clt}, 
\begin{equation*}
 \sqrt{T} (\bar g - \pi[g]) \convdist \N (0, \sigma^2(g)),
\end{equation*}
where 
\begin{equation}\label{eq: resolvent formula for asymptotic variance}
 \sigma^2(g) = \langle g - \pi[g], L^{-1} (g-\pi[g]) \rangle_{\pi}.
\end{equation}
Here, $L^{-1} (g-\pi[g])$ denotes any function in the domain of $L$ with \[L(L^{-1}(g-\pi[g]))=g-\pi[g]\]
and $\pi[L^{-1}(g-\pi[g])]=0$.
Such a function must exist when $g \in L^2(\pi)$ and $X_t$ is reversible~\cite{kipnis1986clt}.

We now show that the process $X^h_t$ meets the conditions above for the central limit theorem.
First, we recall that the generator of $X^h_t$ is the operator
\begin{equation*}
 L_h = \beta^{-1} \Delta -\nabla V \cdot \nabla
\end{equation*}
with domain 
\begin{equation*}
D(L_h) := \{g \in C^2(U^h): \nabla g (x) \cdot \mathbf{n}(x) = 0 \text{ for all } x \in \partial U^h \};
\end{equation*}
see~\cite[Proposition~3.2]{andres2009pathwise} for the case of a convex polyhedron or~\cite[Chapter~8]{ethier1986markov} for a domain with $C^3$ boundary.

By~\cite[Theorem~4.3.3]{jiang2004noneq}, a process $Y_t$ with invariant distribution $\pi$ is reversible if its generator is symmetric and it has the strong continuity property 
\begin{equation}\label{eq: strong continuity of Xht}
 \lim_{t \rightarrow 0^+} \lVert T_t f - f \rVert_\pi = 0 \text{ for all } f \in L^2(\pi),
\end{equation}
where $T_t f(x) := \E_x [f(Y_t)]$ denotes the backwards semigroup associated with $Y_t$.
The generator $L_h$ of $X^h_t$ is symmetric, since for all $f,g \in D(L_h)$, using integration by parts, we have
\begin{align}
-\beta^{-1} \langle \nabla f, \nabla g \rangle_{\pi}
&= -\beta^{-1} \int_{U_h} \nabla f \cdot \nabla g z_h^{-1} \exp(-\beta V) \, dx \nonumber \\
&= 
\beta^{-1} \int_{U_h} f \diver (z_h^{-1} \exp(-\beta V) \nabla g) \,  dx 
\nonumber \\
&\quad -\beta^{-1} \int_{\partial U_h} f z_h^{-1} \exp(-\beta V) \nabla g \cdot n \, dS 
\nonumber \\
&= 
\int_{U_h} (\beta^{-1} \Delta g - \nabla V \cdot \nabla g ) f z_h^{-1} \exp(-\beta V) \,  dx \nonumber \\
&= \langle f, L_h g \rangle_\pi.
\label{eq: form corresponding to generator}
\end{align}
(Here, $z_h^{-1} := \int_{U_h} \exp(-\beta V) \, dx$ is the normalizing constant for $\pi_h$.)
Since $\langle \nabla f, \nabla g \rangle_{\pi}$ is invariant under exchanging $f$ and $g$, $\langle f , L_h g \rangle_\pi = \langle L_h f, g \rangle_\pi$ and $L_h$ is symmetric.
We postpone discussion of the strong continuity of $X^h_t$ to the end of the proof.

We now use the Poincar\'e inequality (Lemma~\ref{lem: poincare inequality}) and~\eqref{eq: form corresponding to generator} to prove that $X^h_t$ is ergodic and to estimate the term $L_h^{-1} (g-\pi_h[g])$ appearing in the formula for $\sigma^2_h(g)$; in essence, we adapt the approach outlined in~\cite[Section~3]{LelievreStoltz:PDEinMD} to the family of reflected processes $X^h_t$.
We prove ergodicity first. 
By~\cite[Proposition~2.2]{bhattacharya1982functionalclt}, a process is ergodic if and only if $0$ is a simple eigenvalue of its generator. 
By the Poincar\'e inequality (Lemma~\ref{lem: poincare inequality}) and~\eqref{eq: form corresponding to generator}, for all $u \in D(L_h)$,
\begin{align}
  \lVert u - \pi_h[u] \rVert_{L^2(\pi_h)}^2
  \leq C_h^2 \lVert \nabla u \rVert_{L^2(\pi_h)}^2 
 = C_h^2 \beta\langle u, -L u \rangle_{\pi_h} 
 \leq C_h^2 \beta \lVert u \rVert_{L^2(\pi_h)}
 \lVert L u \rVert_{L^2(\pi_h)}, \label{eq: coercivity for biased processes}
\end{align}
where 
\begin{equation*}
 C_h = h \Lambda(U)
    \exp\left (\frac{\beta}{2} \left (  \sup_{U_h} V - \inf_{U_h} V \right ) \right ).
\end{equation*}
Now if $u$ is not constant, $\lVert u - \pi_h[u] \rVert_{L^2(\pi_h)}^2 > 0$, so $\lVert L_h u \rVert_{L^2(\pi_h)}> 0$ and $u$ is not an eigenvector with eigenvalue $0$.
Hence, $0$ is a simple eigenvalue of $L_h$, and $X^h_t$ is ergodic.

Finally, we estimate $\sigma^2_h(g)$. 
We have
\begin{equation*}
 \lVert u \rVert_{L^2(\pi_h)} \leq C_h^2 \beta \lVert Lu \rVert_{L^2(\pi_h)}. 
\end{equation*}
Taking $u = L_h^{-1} (g - \pi_h[g])$ in the above yields
\begin{equation*}
 \lVert L_h^{-1} (g - \pi_h[g]) \rVert_{L^2(\pi_h)} \leq
 C_h^2 \beta \lVert g - \pi_h[g] \rVert_{L^2(\pi_h)},
\end{equation*}
which implies 
\begin{equation*}
 \sigma^2_h(g) = \langle g - \pi_h[g], L^{-1} (g-\pi_h[g])\rangle_{\pi_h} \leq C_h^2 \beta \var_{\pi_h}(g), 
\end{equation*}
using the Cauchy--Schwarz inequality.

It remains to show that the process $X^h_t$ has the strong continuity property~\eqref{eq: strong continuity of Xht}. 
We only sketch an argument, since the basic ideas are standard. 
First, one can use the Lipschitz continuity of strong solutions of the reflected process~\cite[Lemma~4.1]{andres2009pathwise} to show that $X^h_t$ has the Feller property.
(That is, one can show that $T_t u$ is continuous whenever $u$ is continuous.)
In addition, since the process $X^h_t$ has an infinitesimal generator, we have the pointwise continuity property
\begin{equation}\label{eq: pointwise continuity at zero}
 \lim_{t \rightarrow 0^+} T_t u(x) = u(x)
\end{equation}
for all $x \in U_h$ and all $u \in D(L_h)$. 
Now we have $\lVert T_t \rVert_\infty \leq 1$ for all $t\geq 0$, where  $\lVert T_t \rVert_{\infty}$ is the operator norm of $T_t$ on the space of continuous functions with the sup-norm, and therefore by a density argument the limit~\eqref{eq: pointwise continuity at zero} holds for all continuous $u$.
Hence, by~\cite[Lemma~1.4]{bottcher2013fellerprimer}, we have 
\begin{equation*}
 \lim_{t \rightarrow 0^+} \sup_{x \in U_h} \lvert T_t u (x)-  u(x) \rvert=0
\end{equation*}
for all continuous $u$.
The strong continuity property~\eqref{eq: strong continuity of Xht} then follows by another density argument, using that $\lVert T_t \rVert_{L^2(\pi_h)} \leq 1$ for all $t \geq 0$.
\end{proof}
\section{Proof of Theorems~\ref{thm: low temp EMUS} and~\ref{sec: small probability limit}}
\label{apx: proof of low temp limit}
\begin{proof}[Proof of Theorem~\ref{thm: low temp EMUS}]
By Corollary~\ref{cor: upper bound, partition of unity}, since the bias functions are a partition of unity, we have 
\begin{align}
  \sigma^2(g) 
  \leq 
  2 \sum_{\iidx \in \mathbb{Z}^d/K\mathbb{Z}^d} 
  \kappa_\iidx^{-1}
  \Bigg \{
  &C(\bar g_\iidx) z_\iidx^2 
   + \var_\pi(g) \tr (R^\iidx)
    \sum_{\substack{ \jidx \neq \iidx \\ F_{\iidx\jidx} >0}}
  \frac{1}{F_{\iidx\jidx}} \Bigg \}.
  \label{eq: restatement of upper bound on variance for low temp}
\end{align}
To prove the desired upper bound, we substitute estimates of $C(\bar g_\iidx)$, $R^\iidx$, and $F_{\iidx\jidx}$ into the inequality above.

First, we consider the asymptotic covariances $R^\iidx$ and $C(\bar g_\iidx)$.
Let 
\begin{equation*}
 h=1/K.
\end{equation*}
The diameter of $U_\iidx$ is $2 \sqrt{d} h$, so by Assumption~\ref{asm: dependence of iat on strata}
\begin{align}
  R^\iidx_{\jidx \jidx}
  &\leq 
    C h^a \beta^b \exp \left ( 2 \sqrt{d} h\beta \lVert \nabla V \rVert_{L^\infty}  \right ) \leq
    C h^{a-b} \exp \left ( 2 \sqrt{d} \lVert \nabla V \rVert_{L^\infty} \right ).
    \label{eq: upper bound on sigma in low temp}
\end{align}
(The second inequality follows since $h\beta \leq 1$ by definition.) 
Similarly, 
\begin{equation}\label{eq: upper bound on tau in low temp}
 \acov (\bar g_\iidx) \leq C h^{a-b} \exp \left ( 2 \sqrt{d} \lVert \nabla V \rVert_{L^\infty} \right ) \var_{\pi_\iidx}(g). 
\end{equation}

Second, by Lemma~\ref{lem: size of entries in F}, the nonzero entries of the overlap matrix $F$ are bounded below as $\beta$ tends to infinity:
\begin{equation} \label{eq: lower bound on F in low temp}
 F_{\iidx\jidx} \geq \frac{\exp \left (- 2 \sqrt{d} \lVert \nabla V \rVert_{L^\infty} \right )}{4^d} 
\end{equation}
for all $\iidx, \jidx$ so that $F_{\iidx, \jidx}>0$.
We also observe that each row of $F$ has $3^d$ nonzero entries, since $F_{\iidx, \iidx+\mathbf{k}} > 0$ only when all entries of $\mathbf{k}$ belong to $\{-1,0,1\}$.

We now estimate the term involving $\acov(\bar g_\iidx)$ in~\eqref{eq: restatement of upper bound on variance for low temp}.
By~\eqref{eq: upper bound on tau in low temp}, we have
\begin{align}
  \sum_{\iidx \in \mathbb{Z}^d/K\mathbb{Z}^d} z_\iidx^2 \acov(\bar g_\iidx) 
  \leq  
    C h^{a-b} &\exp \left ( 2 \sqrt{d} \lVert \nabla V \rVert_{L^\infty} \right ) \sum_{\iidx \in \mathbb{Z}^d/K\mathbb{Z}^d} z_\iidx^2 \var_{\pi_\iidx}(g).
    \label{eq: first estimate of g part of variance in low temp proof}
\end{align}
Now we have 
\begin{align*}
 \var_{\pi_\iidx}(g) 
 &= \pi_\iidx \left [ \lvert g - \pi_\iidx[g] \rvert^2 \right ]
 \leq \pi_\iidx \left [ \left \lvert g - \pi[g] \right \rvert ^2 \right ].
\end{align*}
Therefore,
\begin{align}
\sum_{\iidx \in \mathbb{Z}^d/K\mathbb{Z}^d} z_\iidx^2 \var_{\pi_\iidx}(g) 
&\leq 
\sum_{\iidx \in \mathbb{Z}^d/K\mathbb{Z}^d} z_\iidx \pi_\iidx \left [ \left \lvert g - \pi[g] \right \rvert ^2 \right ]
= \pi \left [ \left \lvert g - \pi[g] \right \rvert ^2 \right ] = \var_{\pi}(g).
\label{eq: second estimate of g part of variance in low temp proof}
\end{align}
(The inequality follows since $0 \leq z_\iidx \leq 1$ for all $\iidx$; the second to last equality follows using~\eqref{eq: final decomposition formula} and that $\{\psi_\iidx\}_{\iidx \in \mathbb{Z}^d/K\mathbb{Z}^d}$ is a partition of unity.)
Thus,
\begin{equation}\label{eq: final estimate of g part of variance in low temp proof}
  \sum_{\iidx \in \mathbb{Z}^d/K\mathbb{Z}^d} z_\iidx^2 \acov(\bar g_\iidx) 
 \leq  
 C h^{a-b} \exp \left ( 2 \sqrt{d} \lVert \nabla V \rVert_{L^\infty} \right ) 
 \var_{\pi}(g).
\end{equation}

It remains to address the term involving $R^\iidx$ in~\eqref{eq: restatement of upper bound on variance for low temp}:
Using~\eqref{eq: upper bound on sigma in low temp},~\eqref{eq: lower bound on F in low temp}, and that each row of $F$ has $3^d$ nonzero entries, we have 
\begin{align}
 \tr (R^\iidx) \sum_{\substack{ \jidx \neq \iidx \\ F_{\iidx\jidx} >0}}
  \frac{1}{F_{\iidx\jidx}} 
  &\leq C 6^{2d} h^{a-b} \exp \left (4 \sqrt{d} \lVert \nabla V \rVert_{L^\infty} \right ) 
  \label{eq: final estimate of sigma part of variance in low temp proof}
\end{align}
for every $\iidx \in \mathbb{Z}^d/K\mathbb{Z}^d$.
Finally, using~\eqref{eq: final estimate of g part of variance in low temp proof},~\eqref{eq: final estimate of sigma part of variance in low temp proof}, and $\kappa_\iidx^{-1} = K^d = \lceil \beta \rceil^d$, we conclude 
\begin{align*}
 \sigma^2(g) 
 &\leq 
   2 C h^{a-b} \Bigg \{ K^d \exp \left ( 2 \sqrt{d} \lVert \nabla V \rVert_{L^\infty} \right ) + 6^{2d} K^{2d} \exp \left (4 \sqrt{d} \lVert \nabla V \rVert_{L^\infty} \right )  \Bigg \} \var_{\pi}(g)\\
&\leq \left ( D \lceil \beta \rceil^{d+b-a} 
+  E \lceil \beta \rceil^{2d+b-a} \right ) \var_{\pi}(g),
\end{align*}
where the constants $D$ and $E$ depend on $d$ and $V$, but not on $g$ or $\beta$.
\end{proof}

We note that if one uses the bias functions proposed in Appendix~\ref{apx: sparse grid}, then the constants $D$ and $E$ in the proof of Theorem~\ref{thm: low temp EMUS} grow only polynomially with the dimension $d$, not exponentially.
However, we do not claim that those bias functions perform better than the uniform grid~\eqref{eq: uniform grid of strata} or the bias functions of Section~\ref{subsec: numerical experiments} in practice.

We now prove Theorem~\ref{thm: small prob EMUS}:

\begin{proof}[Proof of Theorem~\ref{thm: small prob EMUS}] 
  Take $g := \1_{x \geq M}$.
Since the bias functions are a partition of unity, by Corollary~\ref{cor: upper bound, partition of unity}, we have 
\begin{equation}\label{eq: restatement of upper bound on variance for tails}
  \sigma^2_M 
 \leq 
 2 \sum_{i =0}^{K+1} 
 \kappa_i^{-1}
 \left \{
 \acov(\bar g_i)  z_i^2 +
 p_M^2 \tr (R^i)
\sum_{\substack{ j \neq i \\ F_{ij} >0}}
\frac{1}{F_{ij}} \right \}.
\end{equation}

First, we estimate 
\begin{equation*}
  R^i_{jj} \leq C h^a \exp \left (h \max_{x \leq M} \lvert V'(x) \rvert \right ) 
  \leq Ce h^a
\end{equation*}
for all $i=1, \dots, K-1$ by Assumption~\ref{asm: dependence of iat on strata}.
By Assumption~\ref{asm: sampling unbounded strata},
\begin{equation*}
  R^K_{jj} \leq D \text{ for } j = K-1,K,
\end{equation*}
and $R^K_{jj} = 0$ for $j \neq K-1,K$, since $\psi_{K+1}$ is constant over the support of $\pi_K$.
In addition,
\begin{equation*}
  R^{K+1}_{jj} = 0 \text{ for all } j=1, \dots, L,
\end{equation*}
since all bias functions $\psi_i$ take a constant value over the support of $\pi_{K+1}$.
Likewise,
\begin{equation*}
  \acov(\bar g_K) \leq Ce h^a \var_{\pi_K} (g) \leq Ceh^a,
\end{equation*}
and $\acov(\bar g_i) = 0$ for all $i\neq K$.

We now show that the nonzero entries of the overlap matrix are bounded below independent of $M$. 
First, we estimate the entries which are averages over the biased distributions with bounded support.
By Lemma~\ref{lem: size of entries in F}, we have
\begin{equation*}
 F_{ij} \geq \frac{1}{2\exp(2)} > 0
\end{equation*}
for all ${i=0, \dots, K-1}$ and $j$ so that $F_{ij}>0$.
whenever $F_{ij} > 0$.
It remains to address those entries related to biased distributions with unbounded support, so with $i=K,K+1$. 
By Lemma~\ref{lem: lower bound on F in small prob limit}, $F_{K,K+1}$ and $F_{K,K-1}$ are bounded below by some $\theta >0$ independent of $M$, for this choice of bias functions. 
(Lemma~\ref{lem: lower bound on F in small prob limit} and its proof appear in Appendix~\ref{apx: proof of low temp limit}. Lemma~\ref{lem: lower bound on F in small prob limit} is the only part of the proof which relies on Assumption~\ref{asm: assumptions for small prob limit}.)
In addition, for any $i=0,\dots,K+1$, we have $F_{ii}=\frac12$, which implies $F_{K+1,K}=1-F_{K+1,K+1}=\frac12$ since $F$ is stochastic when the bias functions are a partition of unity. 

Finally, we substitute the above estimates of the overlap matrix and the variances into~\eqref{eq: restatement of upper bound on variance for tails}.
Let $c = \min \{\theta, 1/2\exp(2)\}$.
Observe that $h$ decreases with $M$, so $C e h^a \leq E$ for some constant $E$, uniformly in $M$. Let $F=\max\{D,E\}$. 
We have 
\begin{align*}
 \frac{\sigma^2}{p_M^2} &\leq 
  \frac{2}{p_M^2} \sum_{i =0}^{K+1} 
 (K+2)
 \left \{
 \acov(\bar g_i)  z_i^2 +
                          p_M^2 (2F) \frac{2}{c^2} \right \}
                          \leq 
 2(K+2) F  \frac{z_K^2}{p_M^2} +
 \frac{4(K+2)^2}{c^2}.
\end{align*}
We now observe that 
\begin{equation*}
 \frac{z_K}{p_M} = \frac{z_K}{z_{K+1}} = \frac{F_{K+1,K}}{F_{K,K+1}} \leq \frac{1}{c}.
\end{equation*}
Therefore, 
\begin{equation*}
 \frac{\sigma^2}{p_M^2} \leq \frac{2F(K+2) + 4F(K+2)^2}{c^2},
\end{equation*}
which proves the result.
\end{proof}

We now prove Lemma~\ref{lem: lower bound on F in small prob limit}, which is used in the proof of Theorem~\ref{thm: small prob EMUS}.

\begin{lemma}\label{lem: lower bound on F in small prob limit}
Under the hypotheses of Theorem~\ref{thm: small prob EMUS}, there exist constants ${M_1, \theta_+, \theta_- >0}$ depending on $V$ but not on $M$ so that 
\begin{equation*}
 F_{K,K+1} \geq \theta_+ >0 \text{ and } F_{K,K-1} \geq \theta_- >0
\end{equation*}
whenever $M \geq M_1$.
\end{lemma}

\begin{proof}
We consider $F_{K,K-1}$ first. 
We have 
\begin{equation*}
 F_{K,K-1} = \frac12 \frac{\pi([M-h,M))}{\pi([M-h, \infty))} 
 = \frac12 
 \frac{\int_{M-h}^M \exp(-V(x)) \, dx}{\int_{M-h}^\infty \exp(-V(x)) \, dx}.
\end{equation*}
By the integral mean value theorem,
\begin{equation*}
 \int_{M-h}^M \exp(-V(x)) \, dx = h \exp(-V(\xi_{M-h,M}))
\end{equation*}
for some $\xi_{M-h, M} \in [M-h, M]$. 
Moreover, by~\eqref{eq: convexity condition in small prob limit}, we have 
\begin{equation*}
 V(x) \leq V(M) + V'(M) (x-M) \text{ for all } x \geq M \geq M_0.
\end{equation*}
Therefore, when $M -h\geq M_0$,
\begin{align*}
  \int_{M-h}^\infty \exp(-V(x)) \, dx
 &\leq \int_{M-h}^\infty \exp(-V(M-h) - V'(M-h) (x-M+h)) \, dx\\
 &= \frac{\exp(-V(M-h))}{V'(M-h)}.
\end{align*}
It follows that 
\begin{align}
 F_{K,K-1} &\geq h  V'(M-h) \exp(V(M-h)-V(\xi_{M-h,M})) \nonumber \\
 &\geq h  V'(M-h) \exp \left (-h \max_{x \leq M} \lvert V'(x) \rvert \right ) \nonumber \\
 &\geq h V'(M-h)  \exp(-1) \nonumber \\
 &= \frac{V'(M-h)}
 { \left \lceil \max_{x \leq M} \lvert V'(x) \rvert \right \rceil}  \exp(-1),
 \label{eq: quotient of derivatives in small prob lemma}
\end{align}
using the definition $h=M/K$. 

To estimate the quotient in expression~\eqref{eq: quotient of derivatives in small prob lemma}, we distinguish two cases: 
By~\eqref{eq: convexity condition in small prob limit}, $V'$ is nondecreasing on $[M_0, \infty)$, so either ${\lim_{x \rightarrow \infty} V'(x) = C_2 < \infty}$ or ${\lim_{x \rightarrow \infty} V'(x) = \infty}$. 
In the first case, $V'$ is bounded, and we have
\begin{equation}\label{eq: small prob lemma v' bounded case}
 \frac{V'(M-h)}
 { \left \lceil \max_{x \leq M} \lvert V'(x) \rvert \right \rceil} 
 \geq 
  \frac{V'(M_0)}
 { \left \lceil \max_{x \in [0,\infty)} \lvert V'(x) \rvert \right \rceil} 
 >0, 
\end{equation}
whenever $M-h \geq M_0$.
In the second case, for $M$ sufficiently large, 
\begin{equation*}
 \max_{x \leq M} \lvert V'(x) \rvert =  V'(M).
\end{equation*}
Therefore, applying in succession the mean value theorem, the monotonicity of $V'$, assumption~\eqref{eq: ergodicity condition in small prob limit}, and the hypothesis $\lim_{x\rightarrow \infty} V'(x) = \infty$, we have that for all $M$ sufficiently large,
\begin{align}
 \frac{V'(M-h)}{\left \lceil \max_{x \leq M} \lvert V'(x) \rvert \right \rceil} 
 &= \frac{V'(M)}{\lceil V'(M) \rceil} 
 - \frac{V'(M)-V'(M-h)}{\lceil V'(M) \rceil} \nonumber \\
 &\geq \frac{V'(M)}{\lceil V'(M) \rceil}
 -\frac{h V''(\eta_{M-h, M})}{V'(M)} \nonumber \\
 &\geq \frac{V'(M)}{\lceil V'(M) \rceil}
 -\frac{ V''(\eta_{M-h, M})}{V'(\eta_{M-h,M})^2} \nonumber \\
 &\geq \frac{V'(M)}{\lceil V'(M) \rceil}
 -\alpha \nonumber \\
 &\geq \frac{1 - \alpha}{2} \label{eq: small prob lemma v' unbounded case} \\
 &>0. \nonumber
\end{align}
(In the second and third lines above, $\eta_{M-h,M} \in [M-h, M]$ denotes the point guaranteed by the mean value theorem so that ${V'(M)-V'(M-h) = hV''(\eta_{M-h,M})}$.)
It follows from~\eqref{eq: quotient of derivatives in small prob lemma}, \eqref{eq: small prob lemma v' bounded case}, and~\eqref{eq: small prob lemma v' unbounded case} that there exist $M_-,\theta_- >0$ so that 
\begin{equation}\label{eq: small prob proof Fk,k-1 lower bound}
 F_{K,K-1} \geq \theta_- >0
\end{equation}
whenever $M \geq M_-$.

Now we prove that $F_{K,K+1}$ is bounded below. 
We have
\begin{align}
 F_{K,K+1} &= \frac12 
 \frac{\int_M^\infty \exp(-V(x)) \, dx}
 {\int_{M-h}^\infty \exp(-V(x)) \, dx} \nonumber \\
 &= F_{K,K-1}
 \frac{\int_M^\infty \exp(-V(x)) \, dx}
 {\int_{M-h}^M \exp(-V(x)) \, dx} \nonumber \\
 &\geq \theta_-
 \frac{\int_M^{M+h} \exp(-V(x)) \, dx}{\int_{M-h}^M \exp(-V(x)) \, dx} \nonumber \\
 &\geq \theta_-
 \frac{\int_M^{M+h} \exp(V(x-h)-V(x))\exp(-V(x-h)) \, dx}
 {\int_{M-h}^M \exp(-V(x)) \, dx} \nonumber \\
 &\geq \theta_- \exp\left (\min_{[M-h,M+h]} V - \max_{[M-h,M+h]} V \right ) \nonumber \\
 &\geq \theta_- \exp \left (-2h \max_{[M-h,M+h]} \lvert V' \rvert \right ).
 \label{eq: exponent in lower bound lemma small prob section}
\end{align}

As above, to bound the quantity appearing in the exponent in~\eqref{eq: exponent in lower bound lemma small prob section}, we distinguish the two cases $\lim_{x \rightarrow \infty} V'(x) = C_1 < \infty$ and ${\lim_{x \rightarrow \infty} V'(x) = \infty}$. 
In the first case, for $M$ sufficiently large that ${2C_1 \geq \lvert V'(x) \rvert \geq C_1/2}$ whenever ${x \geq M-h}$, we have 
\begin{equation}
 h \max_{[M-h,M+h]} \lvert V' \rvert 
 = \frac{\max_{[M-h,M+h]} \lvert V' \rvert}
 {\left \lceil \max_{[0, M]} \lvert V' \rvert \right \rceil} 
 \leq \frac{2 C_1}{C_1/2} =4.
 \label{eq: lower bound on F small prob limit equation}
\end{equation}
In the second case, for $M$ sufficiently large, 
\begin{align}
  h \max_{[M-h,M+h]} \lvert V' \rvert 
 = \frac{\max_{[M-h,M+h]} \lvert V' \rvert}
 {\left \lceil \max_{[0, M]} \lvert V' \rvert \right \rceil} 
 \leq \frac{V'(M+h)}{V'(M)}.
 \label{eq: lower bound proof quotient of derivatives}
\end{align}
By~\eqref{eq: ergodicity condition in small prob limit}, we have the differential inequality
\begin{equation*}
 V'' < \alpha \lvert V' \rvert^2.
\end{equation*}
This implies 
\begin{equation*}
 V'(M+s) \leq y'(s)
\end{equation*}
for
\begin{equation*}
 y(s) = \frac{1}{V'(M)^{-1} - \alpha s}
\end{equation*}
the solution of the initial value problem 
\begin{equation*}
 y' = \alpha y^2 \text{ and } y(0) = V'(M).
\end{equation*}
Therefore, 
\begin{align*}
 V'(M+h) &\leq \frac{1}{V'(M)^{-1} -\alpha h} =\frac{1}{V'(M)^{-1} -\alpha \lceil V'(M) \rceil^{-1}} \leq \frac{V'(M)}{1-\alpha},
\end{align*}
so by~\eqref{eq: lower bound proof quotient of derivatives}, 
\begin{equation}
 h \max_{[M-h,M+h]} \lvert V' \rvert \leq \frac{1}{1-\alpha}.
 \label{eq: last equation in small prob lower bound on F}
\end{equation}
It follows from~\eqref{eq: exponent in lower bound lemma small prob section}, \eqref{eq: lower bound on F small prob limit equation}, and~\eqref{eq: last equation in small prob lower bound on F} that there exist $M_+,\theta_+ >0$ so that 
\begin{equation}\label{eq: small prob proof Fk,k+1 lower bound}
 F_{K,K+1} \geq \theta_+ >0
\end{equation}
whenever $M \geq M_+$.
\end{proof}

\section{Improved Method of Computing Error Bars}
\label{sec:error-bars}
\newcommand{\avevec}{v}
In \cite[Section VII.B.1]{Thiede2016}, we proposed a practical method of estimating the asymptotic standard deviations (error bars) of averages computed by EMUS. Using the notation established in Appendix~\ref{apx: proof of clt}, our method proceeds as follows:
\begin{enumerate}
  \item Compute $\bar F$, $\{\bar g^\ast_i\}_{i=1}^L$, and $\{\bar 1^\ast_i\}_{i=1}^L$.
\item Compute $w(\bar F)$ and the group inverse $(I-\bar F)^\#$.
  \item Evaluate $\partial_i B$ at $\bar F$, $\{\bar g^\ast_i\}_{i=1}^L$, and $\{\bar \1^\ast_i\}_{i=1}^L$.
  \item Compute the time series
    \begin{equation*}
      \begin{split}
        \bar \zeta^i_t = \partial_i B \cdot &\Big ( \left ( \psi_1(X^i_t), \dots, \psi_L(X^i_t),g^\ast(X^i_t), 1^\ast(X^i_t) \right ) - \left (\bar F_{i1}, \dots, \bar F_{iL},  \bar g^\ast_i, \bar \1^\ast_i \right ) \Big ).
      \end{split}
    \end{equation*}
  \item Compute an estimate $\bar \chi_i^2$ of the integrated autocovariance of
  $\bar \zeta^i_t$ using an algorithm such as ACOR~\cite{acor}.
  \item Compute as an estimate $\sigma^2$ the quantity
  \begin{equation}\label{eq:approx formula for asymptotic variance}
    \bar \sigma^2 := \sum_{i=1}^L \frac{\bar \chi_i^2}{\kappa_i}.
  \end{equation}
\end{enumerate}

We originally proposed computing the group inverse $(I-\bar F)^\#$ using the method of \cite{GolMey:ComputingInvDist} based on the QR factorization. We have since discovered that this method does not always yield sufficiently accurate results. For example, when computing error bars for the marginal in $\mu_2$ in Section~\ref{subsec: numerical experiments}, we observed a highly oscillatory numerical error affecting some entries of $(I-\bar F)^\#$.
That the sign pattern in Figure~\ref{fig:sign-pattern-qr} fails to be symmetric is evidence of this numerical error.
We note that since the exact overlap matrix $F$ is in detailed balance with $w(F)$, we have $\diag(w(F)) F  \diag(w(F))^{-1} = F^\t$.
(Here, $\diag(w(F))$ denotes the diagonal matrix with $w(F)$ along the diagonal.)
Therefore,
\begin{equation*}
  ((I-F)^\#)^\t = \diag(w(F)) (I-F)^\# \diag(w(F))^{-1},
\end{equation*}
which implies that the sign pattern of $(I-F)^\#$ is symmetric since $w(F)$ is positive.
As a result of these numerical errors, we were unable to accurately compute error bars for the EMUS estimate of the marginal density.

\begin{figure}[ht]
  \subfloat[\label{fig:sign-pattern-qr}]{\includegraphics[width=0.5\linewidth]{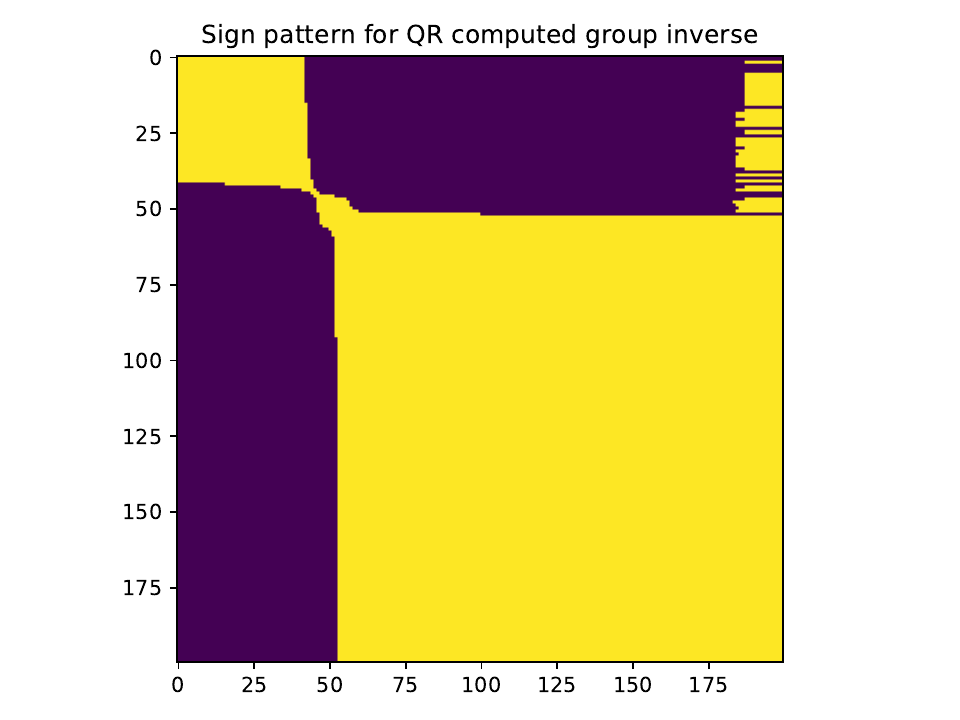}}
\subfloat[\label{fig:sign-pattern-iterative}]{\includegraphics[width=0.5\linewidth]{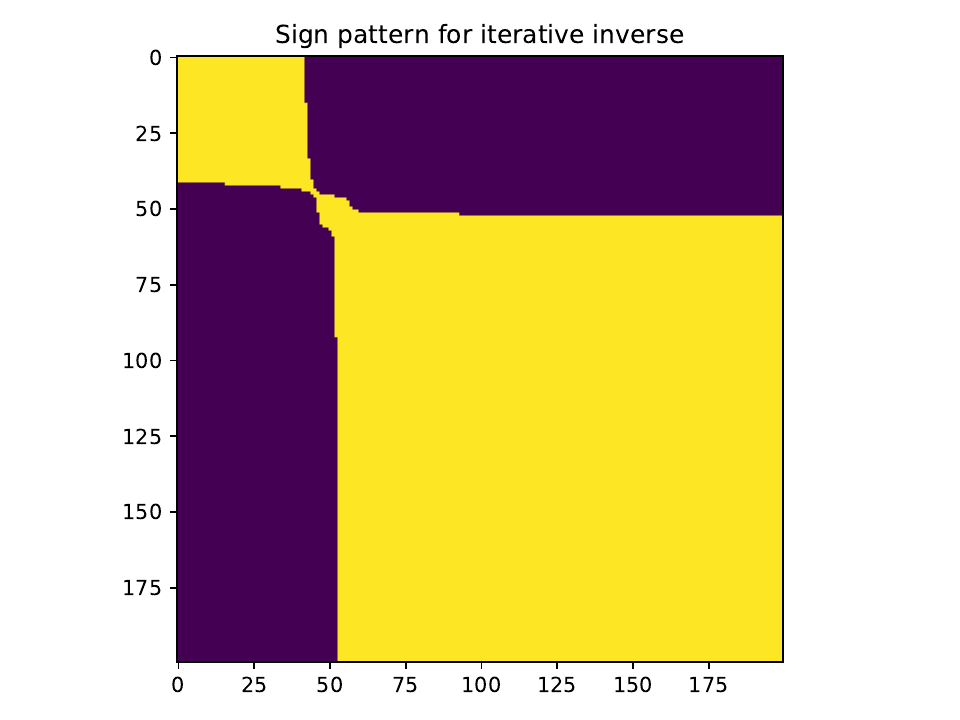}}
\caption{Sign pattern of group inverse $(I-\bar F)^\#$ computed by method of \cite{GolMey:ComputingInvDist} (Figure~\ref{fig:sign-pattern-qr}) and using power iteration (Figure~\ref{fig:sign-pattern-iterative}). Yellow indicates an entry with positive sign, blue a negative sign. Here, we consider the overlap matrix $\bar F$ computed to estimate the marginal density of $\mu_2$ in Section~\ref{subsec: numerical experiments}. The oscillations in sign observed in the upper right corner of Figure~\ref{fig:sign-pattern-qr} are evidence of numerical error.}
\label{fig:qr-group-inverse-error}
\end{figure}

We therefore propose computing the group inverse by a new method combining QR factorization with power iteration. We first compute an estimate $G^0$ of $(I-\bar F)^\#$ by the method of~\cite{GolMey:ComputingInvDist}. We then iterate
\begin{equation}\label{eq:power iteration for group inverse}
  G^{n+1} = \mathscr{I}(G^n)= \tilde F G^n + I-e w(\bar F)^t,
\end{equation}
where $e \in \Real^L$ denotes the column vector of all ones and $\tilde F := (I-e w(F)^t) \bar F$.
We observe that $(I-\bar F)^\#$ is a fixed point of this iteration, since
\begin{align*}
  \mathscr{I}((I-\bar F)^\#)
  &= (I-e w(\bar F)^t) \bar F (I-\bar F)^\# + (I-e \pi(\bar F)^t) \\
  &= (\bar F - I) (I-\bar F)^\# + (I-e w(\bar F)^t) + (I-\bar F)^\# \\
  &=(I-\bar F)^\#.
\end{align*}
Above, we use well known properties of the group inverse, including that the spectral projector $I-e w(\bar F)^t$ commutes with $\bar F$, that $(I-e w(\bar F)^t)(I-\bar F)^\# = (I-F)^\#$, and that $(I-\bar F)(I-\bar F)^\#=I-e w(\bar F)^t$.

Moreover, when $\bar F$ is irreducible, $\mathscr{I}^K$ is a contraction for $K$ sufficiently large.
By the Perron-Frobenius theorem, the spectral radius of $\tilde F$ is smaller than $1-\eps$, for some $\eps >0$.
Therefore, by Gelfand's formula, for any matrix norm $\left \lVert \cdot \right \rVert$, we have ${\lim_{k \rightarrow \infty} \left \lVert \tilde F^k \right \rVert^{1/k} < 1-\varepsilon/2}$, and so for some $K$,
\begin{equation*}
  \left \lVert \tilde F^k \right \rVert < (1 - \varepsilon/2)^k \text{ whenever } k \geq K.
\end{equation*}
Now
\begin{equation*}
  \mathscr{I}^K(G)= \tilde F^K G + (I-e\pi(F)^t) \sum_{j=0}^{K-1} F^j.
\end{equation*}
Thus, assuming that the norm $\lVert \cdot \rVert$ is submultiplicative,
\begin{align*}
  \left \lVert \mathscr{I}^K(G)- \mathscr{I}^K(H) \right \rVert = \left \lVert \tilde F^K (G-H) \right \rVert \leq \left \lVert \tilde F^K \right \rVert \lVert G-H \rVert \leq (1-\varepsilon/2)^K  \lVert G-H \rVert.
\end{align*}
Therefore, the power iteration converges and its limit is the group inverse $(I-\bar F)^\#$.

Using this new method, we computed $(I-\bar F)^\#$ for $\bar F$ the overlap matrix involved in estimating the marginal in $\mu_2$ in Section~\ref{subsec: numerical experiments}. We performed $10^6$ power method iterates. Observe that the sign pattern of the group inverse computed with power iteration is symmetric; see Figure~\ref{fig:sign-pattern-iterative}.   

The power iteration~\eqref{eq:power iteration for group inverse} converges slowly when the spectral gap of $\bar F$ is small. We have shown in \cite{ThVKWe:Perturbation} that the spectral gap may be very small: It decreases exponentially with a temperature parameter in a limit similar to the one analyzed in Section \ref{sec: low temp} above. However, even when the spectral gap is small, we conjecture that a modest number of power iterations will significantly reduce the numerical error in the group inverse, since the error in the initial calculation seems to be highly oscillatory and the power iteration has a smoothing effect. 

\bibliographystyle{spmpsci}
\bibliography{umbrella}
\end{document}